\documentclass[12pt,3p]{article}

\usepackage{amssymb}
\usepackage{graphicx}
\usepackage{amsthm,amssymb,amsfonts,mathrsfs,amsmath}
\usepackage[utf8]{inputenc}
\usepackage[english]{babel}
\newtheorem{theorem}{Theorem}
\newtheorem{corollary}{Corollary}[theorem]
\newtheorem{lemma}[theorem]{Lemma}
\usepackage{epstopdf}
\usepackage{algorithm}
\usepackage{algorithmic}
\usepackage{multicol,subcaption}
\usepackage{comment}
\usepackage{cite}
\usepackage{adjustbox}
\usepackage{pdflscape}
\usepackage{array}
\usepackage{booktabs}
\usepackage[title]{appendix}
\usepackage{tabularx,ragged2e,booktabs}

\usepackage{lineno}
\makeatletter
\def\makeLineNumberLeft{%
  \linenumberfont\llap{\hb@xt@\linenumberwidth{\LineNumber\hss}\hskip\linenumbersep}
  \hskip\columnwidth
  \rlap{\hskip\linenumbersep\hb@xt@\linenumberwidth{\hss\LineNumber}}\hss}
\leftlinenumbers
\makeatother

\begin{document}

\title{Degree Ranking Using Local Information}


%


\author{%
  Akrati Saxena*\\akrati.saxena@iitrpr.ac.in
  \and Ralucca Gera**\\rgera@nps.edu
  \and S. R. S. Iyengar*\\sudarshan@iitrpr.ac.in
  \and *Department of Computer Science and Engineering,\\ Indian Institute of Technology Ropar, India
  }
\date{%
  **Department of Applied Mathematics\\
Naval Postgraduate School, \\Monterey, CA 93943 USA
}

\maketitle
\begin{abstract}
Most real world dynamic networks are evolved very fast with time. It is not feasible to collect the entire network at any given time to study its characteristics. This creates the need to propose local algorithms to study various properties of the network. In the present work, we estimate degree rank of a node without having the entire network. The proposed methods are based on the power law degree distribution characteristic or sampling techniques. The proposed methods are simulated on synthetic networks, as well as on real world social networks. The efficiency of the proposed methods is evaluated using absolute and weighted error functions. Results show that the degree rank of a node can be estimated with high accuracy using only $1\%$ samples of the network size. The accuracy of the estimation decreases from high ranked to low ranked nodes. We further extend the proposed methods for random networks and validate their efficiency on synthetic random networks, that are generated using Erd\H{o}s-R\'{e}nyi model. Results show that the proposed methods can be efficiently used for random networks as well.
\end{abstract}

\section{Introduction}

In complex networks, all nodes have unique characteristics that can be captured using several of the centrality measures proposed in the literature like degree centrality \cite{shaw1954some}, semi-local centrality \cite{chen2012identifying}, closeness centrality \cite{sabidussi1966centrality}, betweenness centrality \cite{freeman1977set}, eigenvector centrality \cite{stephenson1989rethinking}, Katz centrality \cite{katz1953new}, PageRank \cite{brin1998anatomy}, and so on. These centrality measures assign a value to each node based on its importance in the given context. But, in real life applications, we are mostly interested in the relative importance of the node with respect to the top ranked node. It can be measured using the rank of the node based on the given centrality measure. In the present work, we estimate degree rank of a node using local information. The degree of a node $u$ is denoted by $d_u$, which represents the number of neighbors of the node. Degree rank of a node $u$ is defined as, $R_{act}(u) = \sum_{v} X_{uv} + 1$, where $X_{uv} = 1,$ if $d_v > d_u$, otherwise $X_{uv} = 0$. It has been referred as actual degree rank throughout the paper. A node having the highest degree is ranked $1$. All nodes having the same degree will have the same rank.

The classical ranking method collects the degree of all nodes and compares them to compute the rank of a node. If the degree of a node can be computed in $O(1)$ time, then the time complexity of this method is $O(n)$. The space complexity is also high as it requires entire network for the processing. It is not feasible to collect the entire network and store and process it for the large scale networks. Due to this very reason, this method is not feasible for large-scale or distributed real world networks.

Real world networks are highly dynamic, so, the rank of a node keeps on changing with time. To estimate the latest rank of a node, the current snapshot of the entire network is required. Even if we collect this dataset, it might not be useful for further estimations. So, the complexity of pre-processing will be very high. There are some more constraints while studying these networks like online social networks can only be accessed using public interface calls or their API. The number of calls is constant due to API's restrictions. These networks can only be sampled using random walk or its variants like weighted random walk, metropolis-hastings random walk, and so on. It creates the need to propose efficient local algorithms to estimate various properties of the network using less amount of data.

In the present work, we propose four methods to estimate degree rank of a node without having the entire network. These methods use a small snapshot of the network that is collected using different sampling techniques. The proposed methods require some network parameters like network size, maximum, minimum or average degree of the network, that are estimated using random walk samples. Once these pre-processing steps are done, the degree rank of a node can be estimated using the proposed methods. The first method uses power law degree distribution characteristic of real world scale-free networks \cite{barabasi1999emergence} and estimates the degree rank of a node in $O(1)$ time. The next method uses the uniform sampling technique to collect the samples. It computes the local rank of the node in the collected samples, that is extrapolated to estimate its global rank. The last two methods use metropolis-hastings random walk and classical random walk to collect the samples for the rank estimation.

We further study the accuracy and efficiency of the proposed methods on synthetic as well as on real world networks. The accuracy is measured using absolute and weighted error functions. Results show that the proposed methods estimate rank of a node with high accuracy just by using $1\%$ samples of the network size. So, these methods can be used efficiently for online social networks. The proposed methods are verified on 20 real world social networks and the detailed comparison of these methods is discussed in the Results section.

These ranking methods are further extended to rank nodes in random networks. The efficiency of the proposed methods is verified on random networks of 100,000-500,000 nodes. Results show that the degree rank of a node can be estimated efficiently using a small ($1\%$ nodes) sample size.

As per the best of our knowledge, this is the first extensive study of its kind. This work can be helpful to make progress in other domains like identification of influential nodes, comparison of the relative importance of nodes, etc. Identification of influential nodes has been the center of various other research problems like an epidemic \cite{hu2009simulation}, viral marketing \cite{leskovec2007dynamics}, information diffusion \cite{gupta2016modeling, saxena2015understanding}, opinion formation \cite{watts2007influentials}, and so on. It has attracted researchers for quite a long time. The proposed methods can be used to rank influential nodes in different contexts. Degree centrality also has been combined with many other centrality measures to identify the influential nodes \cite{chen2012identifying, hou2012identifying, yu2015node}. Fortunato et al. studied the correlation of in-degree with PageRank of the node \cite{fortunato2006approximating}. They show that the PageRank is directly proportional to the in-degree, modulo an additive constant. Ghoshal and Barabasi studied the dependency of super stable nodes on their degrees \cite{ghoshal2011ranking}. All these applications show the importance of degree ranking in diverse domains of science. As the network size is increasing very fast with time, it is not feasible to implement regular methods. So, local algorithms need to be used in such scenarios as they are practical for large-scale dynamic networks. In \cite{saxena2017afaster}, authors propose fast heuristic methods to estimate the rank of the nodes based on the closeness centrality that itself is a global centrality measure.


The rest of this paper is organized as follows. Next, we discuss related work. In section 3, we discuss methods that are used to estimate the required network parameters. In section 4, all notation that will be used in the paper, are explained. Section 5 describes degree rank estimation methods for scale-free networks. Each of its subsection explains one method in depth. Section 6 explains datasets, error functions, and simulation results for all the proposed methods. Section 7 explains ranking methods for random networks and their validation on Erdos-Renyi networks. The paper is concluded in section 8. This project has various future directions that can be explored further. These are also discussed in the conclusion.

\section{Related Work}

Real world networks are highly dynamic. Their size is increasing very fast with time and in many cases, they are stored in a decentralized way. It is not feasible to store the entire network to study its characteristics like network size, average degree, clustering coefficient, and so on. A small snapshot of the dataset can be collected at any given time to study network characteristics. This has motivated researchers to use sampling techniques to study network parameters. While sampling, the main focus is that the collected dataset should be a good representative of the complete dataset.

The sampling techniques can be mainly categorized as node selection based sampling techniques \cite{leskovec2006sampling}, edge selection based sampling techniques \cite{leskovec2006sampling}, and graph traversal based sampling techniques. In node selection or edge selection methods, nodes or edges are sampled uniformly at random from the network respectively. Haralabopoulos and Anagnostopoulos proposed Enhanced Random Node Sampling method and compared its efficiency with already existing methods \cite{haralabopoulos2014real}. The paper contains the results of the estimation of network parameters like clustering coefficient, average degree, assortativity, and the number of components in real world networks.

The node or edge sampling methods are not feasible in real world networks as the structure of social networks is not known in advance. So, these networks can be sampled using graph traversal techniques like breadth first search (BFS) \cite{even2011graph}, depth first search (DFS) \cite{even2011graph}, forest fire sampling (FFS) \cite{leskovec2006sampling}, snowball sampling \cite{goodman1961snowball}, or random walk based methods like simple random walk (RW) \cite{lovasz1993random}, Metropolis Hastings random walk (MHRW) \cite{metropolis1953equation}, reweighted random walk (RWRW) \cite{hansen1943theory}, respondent driven sampling (RDS) \cite{salganik2004sampling}, supervised random walk \cite{backstrom2011supervised}, Modified TOpology Sampling (MTO) \cite{zhou2016faster}, walk-estimate \cite{nazi2015walk}, Frontier sampling (m-dimensional random walk) \cite{ribeiro2010estimating}, Rank Degree sampling based on edge selection \cite{voudigari2016rank}, preferential random walk \cite{davis2016marginal}, and so on.

Next, we discuss estimation methods for network parameters using sampling methods.
Kurant et al. proposed a method called SafetyMargin that uses Induced Edges sampling techniques to estimate the network size \cite{kurant2012graph}. The proposed method outperforms state of the art methods even using 10 times small sample size. In 2013, Hardiman and Katzir proposed a more efficient method to estimate the network size using random walk samples \cite{hardiman2013estimating}. This method is discussed in more detail in Section 3.1. They also proposed methods to compute average clustering coefficient and global clustering coefficient of the network using random walk. There have been proposed some more methods to estimate network size like \cite{cem2016estimation, musco2016ant, lucchese2015networks, chen2016estimating, ye2011estimating}.


Sampling techniques also have been used to identify degree related properties like high degree nodes, average degree, or degree distribution of the network. Cooper et al. proposed a biased random walk method to identify high degree nodes in the scale-free networks \cite{cooper2012fast}. Marchetti-Spaccamela proposed a method to estimate the degree of a node in directed network \cite{marchetti1988estimate}. Dasgupta et al. proposed a method to estimate the average degree of the network using smooth random walk, that is discussed in Section 3.2 \cite{dasgupta2014estimating}. Eden et al. proposed an algorithm to estimate the average degree using $\tilde{O}(1)$ queries \cite{eden2016sublinear}. There have been proposed some more methods to estimate the average degree \cite{cem2016estimation, lu2012sampling}. Cem and Sarac proposed methods to estimate the size and the average degree of online social networks where only one random neighbor of the node can be accessed using API calls \cite{cem2015estimating}. They further used ego-centric sampling and showed that the use of neighborhood information is not always beneficial to estimate network parameters like network size and average degree \cite{cem2016average}.

Ribeiro et al. studied the mean square error while computing the degree distribution of the network \cite{ribeiro2012estimation}. They further compute the normalized mean square error for estimating the out-degree and in-deg distribution of the directed networks \cite{ribeiro2012sampling}. The proposed method uses Directed Unbiased Random Walk (DURW) that takes a random jump with a fixed probability depending on the degree of the node while taking the walk. The results show that the out-degree distribution can be estimated more efficiently and accuracy of the in-degree distribution is very less unless the graph is not symmetric.

Thus, we have seen that the sampling methods can be used to estimate various network parameters. In the present work, we use sampling techniques to estimate degree centrality rank of the nodes. 

\section{Estimation of Network Parameters}

Different ranking methods require different network parameters that need to be estimated during the preprocessing steps. In this section, we discuss the methods that are used to estimate these parameters.

\subsection{Estimate the Network Size}

The network size is estimated using a method (HK method) that was proposed by Hardiman and Katzir \cite{hardiman2013estimating}. The proposed estimator is based on the concept of collision to count total number of nodes, where samples are collected using the classical random walk. Authors use neighbors' information of the sampled nodes to detect the collision a step before it actually occurs. There is a high probability of collision on short distances due to the local traversal. So, a pair of nodes in random walk samples is considered to count the collision if their distance is long during the random walk. We have considered a pair if their distance is more than $2.5\%$ of sample size. It is the same as taken by the authors.

\subsection{Estimate the Average degree}

The average degree of the network is used while estimating the rank of a node using power law degree distribution. It is estimated using a method (AD method) that was proposed by Dasgupta et al. \cite{dasgupta2014estimating}. The samples are collected using smoothed random walk with a distribution $D_{d,c}$, where the probability to sample a node $u$ is directly proportional to $d_u + c$, and $c$ is constant during the entire random walk. The samples generated using this distribution are equivalent to samples generated from the network, where $c/2$ self-loops are added to each node.

They proposed two estimators called: 1. Guess$\&$Smooth, 2. Smooth. The optimal value of $c$ is decided using Guess$\&$-Smooth estimator. This smoothing parameter $c$ is used by Smooth estimator to collect the samples. These two estimators are combined to propose an estimator that takes $O(logU \cdot loglogU)$ samples to estimate average degree with high accuracy, where $U$ is an upper bound on the maximum degree of the network. Thus, this estimator takes very less number of samples to estimate the average degree.

The estimated average degree of the network is required while estimating the rank of a node using power law degree distribution method. 


\subsection{Estimate the Maximum Degree}

In power law degree distribution, the frequency of highest degree node is almost 1. In the analysis, maximum degree is estimated as the available maximum degree in the samples, $d'_{max}= max\left\lbrace d_u, \forall u \in S\right\rbrace$, where $S$ is the set of samples.

\subsection{Estimate the Minimum Degree}

In real world networks, we observe that the minimum degree is 1 or close to 1. We use the same value of minimum degree for the analysis, so $d'_{min} = 1$. In BA and ER networks, the minimum degree is estimated as the available minimum degree in the samples. 

The minimum and maximum degree are estimated using the network size estimator samples.

\section{Notation}\label{sec:notation}

$\mathcal{G}(f)$ represents a set of networks having $n$ nodes, and all networks are generated using the same degree distribution $f$. Table~\ref{my-label12} contains all notation used in the paper.

\begin{table}[htp]
\centering
\caption{Notation}
\label{my-label12}
\resizebox{\columnwidth}{!}{%
\begin{tabular}{|l|l|}
\hline
\textbf{Notation} & \textbf{Description} \\ \hline
$G$ & A network, $G \in \mathcal{G}(f)$  \\ \hline
$n$ & Total number of nodes in the network \\ \hline
$m$ & Total number of edges in the network \\ \hline
$n'$ & Estimated number of nodes in the network \\ \hline
$n_j$ & Total number of nodes having degree $j$ in the network \\ \hline
$u,v,w$ & Nodes in the network  \\ \hline
$d_u$  & Degree of node $u$  \\ \hline
$d_{max}$  & Maximum degree in the network  \\ \hline
$d_{min}$  & Minimum degree in the network  \\ \hline
$d_{avg}$  & Average degree of the network  \\ \hline
$S$ & Set of sampled nodes  \\ \hline
$s$ & Sample size, $s=|S|$ \\ \hline
$d'_{max}$  & Estimated maximum degree/maximum degree in $S$  \\ \hline
$d'_{min}$  & Estimated minimum degree/minimum degree in $S$ \\ \hline
$d'_{avg}$  & Estimated average degree/average degree of $S$  \\ \hline
$R_{act}(u)$  & Actual rank of node $u$ in the network \\ \hline
$R_{est}(u)$  & Estimated rank of node $u$ in the network  \\ \hline
$R_{local}(u)$  & Rank of node $u$ in sample $S$ \\ \hline
$R_{G}(u)$  & A random variable that denotes the rank of node $u$ in $G$ \\ \hline
$R_{S}(u)$  & A random variable that denotes the rank of node $u$ in $S$ \\ \hline
\end{tabular}
}
\end{table}

\section{Estimate the Degree Rank}

In this section, we will explain four methods to estimate degree rank of a node using local information. 

\subsection{Using Power Law Degree Distribution (PL Method)}

In 1999, Barabasi and Albert observed that degree distribution $f$ of real world scale-free networks follows power law \cite{barabasi1999emergence}. The probability $f(j)$ of a node having degree $j$ is given as $f(j) = cj^{-\gamma}$, where $c$ and $\gamma$ are constants for a network. Due to the power law characteristic, only a few nodes manage to get the higher degree in the network. In real world scale-free networks, the range of the exponent is $2 < \gamma < 3$. The degree rank of a node can be computed using power law equation if its parameters are known. In this section, we propose a method to estimate these parameters that can be used further to estimate degree rank of the node.


\begin{theorem}
In a scale-free network $G$ $(G \in \mathcal{G}(f))$, the power law exponent of degree distribution can be computed as, $\gamma \approx 2 + \frac{d_{min}}{d_{avg}-d_{min}}$, where $d_{min}$ and $d_{avg}$ represent minimum and average degree of the network respectively.
\end{theorem}

\begin{proof}
%

Let network $G$ follows power law degree distribution $f(j)= cj^{-\gamma}$.
First, we derive an equation to estimate the value of $c$. The sum of probabilities of a node having degree $j$ ($d_{min} \le j \le d_{max}$) is equal to $1$. The probability function of degree distribution can be written as,
\begin{center}
\begin{align*}
\sum_{j=d_{min} }^{d_{max} } f(j) = 1.
\end{align*}
\end{center}
We switch to integration\footnote{Here, discrete probability values are considered as continuous probability density function, as this introduces a very small error.} to compute $c$: 

\begin{align*}
\int_{d_{min} }^{d_{max} } f(j) dj &= 1 ,\\
\int_{d_{min} }^{d_{max} } c \cdot j^{-\gamma} dj &= 1.
\end{align*}
After integration, we obtain the value for $c$ to be
\begin{equation*}
c \cdot \frac{ (d_{max})^{1-\gamma} - (d_{min})^{1-\gamma} }{ 1 - \gamma} = 1\\
\end{equation*}
\begin{equation}\label{eq:c}
c =  \frac{1- \gamma}{(d_{max})^{1-\gamma} - (d_{min})^{1-\gamma}}.
\end{equation}

\noindent
To compute $\gamma$, the average degree of the network, $(d_{avg})$, is used. Using $f(j)= c \cdot j^{-\gamma}$, it can be computed as

\begin{align*}
d_{avg} &= \sum_{j=d_{min} }^{d_{max} } j \cdot f(j)\\
d_{avg} &= \int_{d_{min} }^{d_{max} } j \cdot \left( c \cdot j^{-\gamma}\right) dj.
\end{align*}

\noindent
After integration, we have that
\begin{center}
$d_{avg} = c  \cdot \frac{ d_{max}^{2- \gamma} - d_{min}^{2- \gamma}}{2 - \gamma}.$
\end{center}

Putting value of $c$ from equation \eqref{eq:c} in this equation,
\begin{align*}
d_{avg} &= \frac{1- \gamma}{2 - \gamma} \cdot \frac{ d_{max}^{2-\gamma} - d_{min}^{2-\gamma}}{ d_{max}^{1-\gamma} - d_{min}^{1-\gamma}}\\
d_{avg} &= \frac{\gamma -1}{\gamma-2} \cdot \frac{ d_{max}^{\gamma-2} - d_{min}^{\gamma-2}}{ d_{max}^{\gamma-1} - d_{min}^{\gamma-1}} \cdot d_{max} \cdot d_{min}
\end{align*}
where, $d_{min} << d_{max}$, and $2 < \gamma < 3$ for scale-free real networks~\cite{barabasi1999emergence}.

\begin{center}
$d_{avg} \approx \frac{\gamma-1}{\gamma-2} \frac{d_{max}^{\gamma-2}}{d_{max}^{\gamma-1}} \cdot d_{max} \cdot d_{min}$\\
\end{center}
\begin{center}
$d_{avg} \approx \frac{\gamma-1}{\gamma-2} \cdot d_{min}$\\
\end{center}
i.e. $ \gamma \approx 2 + \frac{d_{min}}{d_{avg}-d_{min}}$.
\end{proof}

We next present the expected degree rank of a node. 

\begin{theorem}\label{expected}
In a network $G$ $(G \in \mathcal{G}(f))$, the expected degree rank of a node $u$ can be computed as, $E[R_{G}(u)] \approx n \left( \frac{d_{max}^{1-\gamma} - (d_{u}+1)^{1-\gamma}}{ d_{max}^{1-\gamma} - d_{min}^{1-\gamma} } \right)  + 1,$ where $\gamma$ is the power law exponent of the degree distribution of network $G$.
\end{theorem}

\begin{proof}

In a given network $G$, the actual rank of a node $u$ having degree $d_u$ can be computed as,

\begin{center}
\begin{align*}
R_{act}(u) = \sum_{j=d_{u}+1}^{d_{max}}n_j  + 1
\end{align*}
\end{center}
where $n_j$ represents total number of nodes having degree $j$ in network $G$.  Let  $N_j$ be a random variable that represents total number of nodes having degree $j$ in $G$. 
Then, the expected value of $N_j$ can be computed as, $E[N_j]=n\cdot f(j)$.
Thus  the expected degree rank of a node $u$ can be computed as
\begin{align*}
E[R_{G}(u)] &= E\left[ \sum_{j=d_{u}+1}^{d_{max}}N_j  + 1\right] \\
E[R_{G}(u)] &= \sum_{j=d_{u}+1}^{d_{max}}E[N_j]  + 1\\
E[R_{G}(u)] &= \sum_{j=d_{u}+1}^{d_{max}}n\cdot f(j) + 1\\
E[R_{G}(u)] &\approx n \int_{d_{u}+1}^{d_{max}}f(j) dj +1.
\end{align*}

\noindent
Since  $f(j) = cj^{-\gamma}$, after the integration of
$E[R_{G}(u)] \approx n \int_{d_{u}+1}^{d_{max}} c \cdot j^{-\gamma} dj+1$
 we have
$$E[R_{G}(u)] \approx n c \frac{d_{max}^{1-\gamma} - (d_{u}+1)^{1-\gamma}}{ {1 - \gamma} } + 1.$$

\noindent
Replacing the value of $c$ from equation \eqref{eq:c}, we obtain
$$E[R_{G}(u)] \approx n( \frac{d_{max}^{1-\gamma} - (d_{u}+1)^{1-\gamma}}{ d_{max}^{1-\gamma} - d_{min}^{1-\gamma} } ) + 1$$
as desired.~\end{proof}

And so, using Theorem~\ref{expected} and given general estimators about the network, we can estimate the degree rank of nodes.
\begin{corollary}\label{estimated}
In a network $G$ $(G \in \mathcal{G}(f))$, the degree rank of a node $u$ can be estimated as, $$R_{est}(u) = n'\left(\frac{{(d'_{max})}^{1-\gamma} - (d_{u}+1)^{1-\gamma}}{ {(d'_{max})}^{1-\gamma} - {(d'_{min})}^{1-\gamma} } \right)  + 1,$$ where $\gamma = 2 + \frac{d'_{min}}{d'_{avg}-d'_{min}}$, and  $n', d'_{min}, d'_{max}$, and $d'_{avg}$ denote the estimated value of network size, minimum degree, maximum degree, and average degree of the network respectively.  
\end{corollary}

In one of our previous works, we have validated this method on BA networks \cite{saxena2015rank, saxena2016estimating}. The proposed method estimates the rank with high accuracy for BA networks but does not give good results for real world networks, as they follow power law degree distribution with a droop head and a heavy tail. We further compute variance in degree rank estimation using power law degree distribution and the results show that there is a high variance for lower degree nodes \cite{saxena2015estimating}. 

Next, we propose few more sampling based approaches that perform better on real world networks. These are discussed below.

\subsection{Using Uniform Sampling (US Method)}

%

In this section, the uniform sampling technique is used to collect a small sample of actual dataset. In uniform sampling, the probability of sampling a node is equal to $1/n$, where $n$ is total number of nodes. Uniform samples preserve the characteristics of actual dataset. So, the collected samples follow similar degree distribution as observed in real world networks. Here we assume that the network $G$ is generated using degree distribution $f_1$ and $G \in \mathcal{G}(f_1)$. Now, Theorem 3 can be used to estimate the rank of a node using uniform samples. The expected global rank of a node can be estimated by extrapolating its local rank in the collected sample set.

\begin{theorem}
In a network $G$ $(G \in \mathcal{G}(f_1))$, if sample $S$ is collected uniformly, the expected local rank of node $u$ can be computed as, $E[R_{S}(u)] \approx \frac{s}{n} E[R_{G}(u)]$, where $R_{G}(u)$ and $R_{S}(u)$ are random variables that denote the rank of node $u$ in network $G$ and sample $S$ respectively.
\end{theorem}

\begin{proof}
We are interested in computing the rank of a node $u$ having degree $d_u$. Let's take a random variable $N_j$, that denotes the number of nodes having degree $j$ in the network. The expected value of $N_j$ can be computed as, $E[N_j] = n \cdot f_1(j)$.

The expected rank of node $u$ in network $G$ can be computed as:
\begin{center}
$E[R_{G}(u)] = E[\sum_{j=d_u+1}^{d_{max}}(N_j) +1]$
\begin{equation}\label{eq:ract}
E[R_{G}(u)] = \sum_{j=d_u+1}^{d_{max}}(n \cdot f_1(j)) +1
\end{equation}
\end{center}
Now, we have a uniform sample $S$ of size $s$. In network $G$, the probability $p$ to sample a node $y$ uniformly at random having degree greater than $d_u$ $(d_y > d_u)$ can be defined as,

\begin{center}
$p = \frac{\sum_{j=d_u+1}^{d_{max}}(n \cdot f_1(j))}{\sum_{j=1}^{d_{max}}(n \cdot f_1(j))}$ 
\end{center}
Using equation \eqref{eq:ract},
\begin{center}
$p = \frac{E[R_{G}(u)] -1}{\sum_{j=1}^{d_{max}}(n\cdot f_1(j))} $\\
$ p = \frac{E[R_{G}(u)] -1}{n \sum_{j=1}^{d_{max}}( f_1(j))} $
\end{center}
Using the property of probability distribution $\sum_{j=1}^{d_{max}}f_1(j)=1$.
\begin{center}
\begin{equation}\label{eq:p}
E[R_{G}(u)] = p \cdot n +1
\end{equation}
\end{center}
The expected value of local rank of node $u$ in sample $S$ can be computed as,
\begin{center}
$E[R_{S}(u)] = \sum_{j=0}^{s} \begin{pmatrix} s\\ j \end{pmatrix} p^j (1-p)^{(s-j)} j + 1$
\end{center}
\begin{center}
\begin{equation}\label{eq:l}
E[R_{S}(u)] = s \cdot p +1
\end{equation}
\end{center}
Using equations \eqref{eq:p} and \eqref{eq:l},\\
\begin{center}
$E[R_{S}(u)] = \frac{s}{n} E[R_{G}(u)] + \frac{n-s}{n}$
\end{center}
Where, $0 \leq (n-s)/n <1$, if $s \leq n$. So,
\begin{center}
$E[R_{S}(u)] \approx \frac{s}{n} E[R_{G}(u)] $
\end{center}
\end{proof}

In a network G, $R_{act}(u) \approx E[R_{G}(u)] $ and $ R_{local}(u) \approx E[R_{S}(u)]$. $R_{local}(u)$ denotes the rank of node $u$ in sample $S$, and $R_{local}(u)= \sum_{j=i+1}^{d'_{max}}(n'_j) +1 $, where $n'_j$ is the number of nodes having degree $j$ in sample $S$. Using theorem 3, the actual rank of node $u$ can be computed as,
\begin{center}
\begin{equation}\label{eq:es}
R_{act}(u) \approx \frac{n}{s}R_{local}(u)
\end{equation}
\end{center}

\begin{corollary}
In a network $G$, using uniform samples, degree rank of a node $u$ can be estimated as, $R_{est}(u) = \frac{n'}{s}R_{local}(u)$, where $n'$ is the estimated network size.
\end{corollary}

\subsection{Using Metropolis-Hastings Random Walk (MH Method)}

In most of the online networks, uniform sampling is not possible as node ids are not known well in advance. These networks can be sampled using graph sampling techniques like breadth first traversal, random walk, etc. These sampling methods are biased towards higher degree nodes and fail to generate uniform samples. In this method, we use metropolis-hastings random walk that generates sample equivalent to uniform samples, that can be used for rank estimation.


\textbf{Metropolis-Hastings Random Walk:} This technique was first proposed by Metropolis et al. \cite{metropolis1953equation} in 1953. In this method, the probability distribution of random walk is modified so that the collected samples retain the properties of the actual distribution of the dataset. At each time step, the crawler will move to the next node with probability $p$ and will stay at the same node with probability $(1-p)$. So, the probability distribution can be modified as,

$P_{u \rightarrow v} =\left\{\begin{matrix}
\frac{1}{d_u} \cdot min(1,\frac{d_u}{d_v}),  & if \; v \; is \; the \; neighbor \; of \; u, \\
1-\sum_{w \neq u}P_{u \rightarrow w}, & if \; v=u, \\
0, & otherwise.
\end{matrix}\right.$

This probability distribution collects more samples of lower degree nodes and fewer samples of higher degree nodes, so the collected samples are not biased towards higher degrees. Gjoka et al. studied that in real world network, the samples collected using metropolis-hastings random walk are equivalent to uniform samples, and can be used to study the network parameters \cite{gjoka2010walking}. Corollary 3.1 can be used to estimate degree rank using MHRW samples.

%

\subsection{Using Random Walk (RW Method)}

The classical random walk is a well known easier method to collect the samples in large dynamic networks. In \textbf{Random Walk}, a crawler starts from a randomly chosen node. It moves to the next node that is chosen uniformly at random among the neighbors of the current node \cite{lovasz1993random}. The probability to move to node $v$ from node $u$ is defined as,

$P_{u \rightarrow v} =\left\{\begin{matrix}
\frac{1}{d_u}, & if \; v \; is \; a \; neighbor \; of \; u, \\
0, & otherwise.
\end{matrix}\right.$

In a random walk, the probability of a node being sampled converges to a stationary distribution, $p(u)=d_u/2m$. So, the collected samples are biased towards high degree nodes. We propose Theorem~\ref{RW} to estimate degree rank using random walk samples.

First, notice that in a random walk, the probability of a node being sampled is directly proportional to its degree. These samples can be converted to uniform samples using a new probability distribution, where the probability of picking a node is inversely proportional to its degree $p(u) \propto 1/d_u$,  known as re-weighted random walk sampling technique~\cite{hansen1943theory}. 

\begin{theorem}~\label{RW}
In a network $G$ $(G \in \mathcal{G}(f_1))$, using random walk sample $S$, the degree rank of node $u$ can be computed as, $R_{act}(u) \approx \frac{n}{k} \cdot R_{local}(u)$, where $R_{local}(u) = \sum_{j=d_u+1}^{d'_{max}}(q(j) \cdot k) + 1$, and $k$ is a constant, $q(j)$ is the re-sampling probability function  $q(j)=\frac{n'_j/j}{\sum_{i=d'_{min}}^{d'_{max}} n'_i/i}$, and $n'_j$ represents total number of nodes having degree $j$ in sample $S$.
\end{theorem}

\begin{proof}
The probability $q$ to resample a $j$ degree node can be computed as, $q(j)= \frac{n'_j/j}{\sum_{i=d'_{min}}^{d'_{max}} n'_i/i}$, where $n'_j$ represents total number of nodes having degree $j$ in sample $S$.

To estimate the degree rank of a node, collect $q(j) \cdot k$ samples of each degree $j$ from $S$, where $k$ is a constant. So, total number of new sampled nodes $|S'| = \sum_{j=d'_{min}}^{d'_{max}}\left( q(j) \cdot k \right) =k$.
Then, the rank of node $u$ in $S'$ can be computed as, $R_{local}(u) = \sum_{j=d_u+1}^{d'_{max}}(q(j) \cdot k) + 1$.


As the new sample set $S'$ follows uniform distribution, the rank of node $u$ can be computed using equation \eqref{eq:es}, $R_{act}(u) \approx \frac{n}{k} \cdot R_{local}(u)$.~\end{proof}

In the experiments, the value of $k$ is chosen as $k = min(1/q)$, so that the regenerated samples also contain higher degree nodes and their rank is estimated with high accuracy. 

\begin{corollary}
In a network $G$, using random walk samples, the degree rank of a node $u$ can be estimated as, $$R_{est}(u)= n' \cdot \frac{\sum_{j=d_u+1}^{d'_{max}} \left( \frac{n'_j}{j} \cdot k \right) }{\sum_{j=d'_{min}}^{d'_{max}} \left( \frac{n'_j}{j} \cdot k \right) },$$ where $n'_j$ represents total number of nodes having degree $j$ in sample $S$, and $k$ is a constant.
\end{corollary}

\section{Simulation Results}

In this section, we will discuss the datasets, error functions, and simulation results.

\subsection{Datasets}

All proposed methods are simulated on both synthetic as well as on real world scale-free networks. Synthetic networks are generated using Barab\'{a}si-Albert (BA) model $G(n,k)$, where each new coming node makes $k$ preferential connections with already existing nodes \cite{barabasi1999emergence}. The probability $p(u)$ to make a connection with an existing node $u$ is directly proportional to the degree of node $u$, as $p(u) = {d_u}/{\sum_{v}d_v}$. So, the nodes having higher degrees acquire more links over time and it gives birth to power law degree distribution. All synthetic datasets are explained in Table~\ref{syn-data}.

\begin{table}[htp]
\centering
\caption{Datasets}
\label{syn-data}
\begin{tabular}{|l|l|l|l|l|}
\hline
Network & $\#$Nodes & $\#$Edges  \\ \hline 
BA1 & 100000 & 999900  \\ \hline
BA2 & 200000 & 1999900 \\ \hline
BA3 & 300000 & 2999900  \\ \hline
BA4 & 400000 & 3999900  \\ \hline
BA5 & 500000 & 4999900  \\ \hline
\end{tabular}
\end{table}

All real world datasets are explained below:

\begin{enumerate}

\item Academia Online Social Network: Academia.edu is an online website where academics share research papers. This is an extracted social network from this website \cite{nr, Fire2011}. It contains 200167 nodes and 1022440 edges.

\item Actor Collaboration Network: In actor collaboration network, nodes represent actors and they are connected by an edge if they both have performed in the same movie \cite{barabasi1999emergence}. This network contains 374511 nodes and 15014839 edges.

\item Catster Friendship Network: This social Network is created using the friendships between catster.com website users \cite{catster}. catster provides a platform to cat owners and lovers, where they can connect with each other and share the information. It contains 148826 nodes and 5447464 edges.

\item DBLP Collaboration Network: This is a coauthorship network extracted from DBLP computer science bibliography, where edge denotes that the authors have common publications \cite{yang2015defining}. This network contains 317080 nodes and 1049866 edges.

\item Delicious online social network: Delicious is a social bookmarking web service for storing, sharing, and discovering web bookmarks. The dataset contains all links among users \cite{nr, zafarani2014users}. It contains 536108 nodes and 1365961 edges.

\item Digg Friendship Network: This friendship network was extracted from Digg website in 2009 \cite{hogg2012social}. It contains 261489 nodes and 1536577 edges.

\item Dogster Friendship Network: This social Network is created using the friendships between users of the website http://www.dogster.com \cite{nr}. catcher provides a platform to cat owners and lovers, where they can connect with each other and share the information. It contain 426485 nodes and 8543321 edges.

\item Douban online social network: Douban is an online social network that provides user review and recommendation services for movies, books, and music. This is the friendship network extracted from the website \cite{nr, zafarani2014users }. It contains 154908 nodes and 327162 edges.

\item European Email Communication Network: This is the email communication network of a European research institution, where a node represents an individual person and an edge represents that at least one email has been exchanged between them \cite{leskovec2007graph}. This dataset was collected from October 2003 to May 2005 (18 months).

\item Facebook Network: Facebook is the most popular online social networking site today. This dataset is the induced subgraph of Facebook, where users are represented by nodes and friendships are represented by edges \cite{Traud_2011fs,traud2012social}. It contains 3097165 nodes and 23667394 edges.

\item Friendster Network: This is the induced subgraph of Friendster online social network \cite{nr}. Nodes represent users and a directed edge (a,b) indicates that user a has added user b to his friendship lists. The network is converted to a undirected network for study. It contains 5689498 nodes and 14067887 edges.

\item Foursquare Network: Foursquare is a location-based social networking software for mobile devices that can be accessed using GPS. This dataset is an induced subgraph of friendships of Foursquare \cite{zafarani2009social}. It contains 639014 nodes and 3214985 edges.

\item Gowalla Social Network: This network is extracted from a location-based social network called, Gowalla \cite{cho2011friendship}. This was used to share the locations among its users. In this network, a node represents a user and an edge indicates the friendship between the user. It contains 196591 nodes and 950327 edges.

\item Google Plus Social Network: This is an induced subgraph of Google plus online social network \cite{mcauley2012learning}. It contains 107614 nodes and 12238285 edges.

\item Hollywood Collaboration Network: This is an undirected collaboration network of Hollywood movie actors where nodes are actors, and there is an edge between two actors if they have appeared in a movie together \cite{BoVWFI}.It contains 1069126 nodes and 56306653 edges.

\item Hyves: Hyves is a popular social networking site in the Netherlands that is mainly used by Dutch visitors. This is the induced subgraph of this that was collected in December 2010 \cite{zafarani2009social}. The network is undirected and unweighted. It contains 1402673 nodes and 2777419 edges.

\item last.fm Network: Last.fm is a music website that has more than 40 million active users \cite{nr, konstas2009social }. This is the induced friendship networks of the bloggers from this website. It contains 1191805 nodes and 4519330 edges.

\item Livemocha Network: Livemocha was an online language learning website and this network is extracted from the social connections of the website \cite{ZafaraniLiu2009}. This contains 104103 nodes and 2193082 edges.

\item Pokec Online Social Network (soc-pokec): Pokec is a popular online social network in Slovakia. The dataset contains a list of user relationships \cite{nr}. It contains 1632803 nodes and 22301964 edges.

\item Youtube Social Network: This is the induced subgraph of Youtube social network \cite{ zafarani2009social}. This network contains 1134885 nodes and 2987468 edges.
\end{enumerate}

\subsection{Error Functions}
The accuracy of all methods is evaluated using absolute and weighted error functions. These are discussed below:
\begin{enumerate}
\item Absolute Error: Absolute error for a node $u$ is computed as,
\begin{center}
$Err_{abs}(x) = |R_{est}(u) - R_{act}(u)|$
\end{center}

The \textbf{percentage average absolute error} can be computed as $$ Err_{paae} = \frac{\text{average \; absolute \; error}}{\text{network \; size}} \cdot 100\%.$$

\item Weighted Error: In real life applications, the significance of the error depends on two important parameters: $1.$ rank of the node, and $2.$ network size. The same rank difference has more impact for the higher ranked nodes than the lower ranked nodes. Similarly, the same error in the rank will be perceived higher in smaller networks than the larger networks. We consider both of these parameters and propose a weighted error function. It is defined as,
\begin{center}
$Err_{wtd}(x) = \frac{Err_{abs}(x)}{n} \cdot \frac{(n - R_{act}(u) +1)}{n} \cdot 100 \%$
\end{center}

Where, $\frac{(n - R_{act}(u) +1)}{n} \times 100$ denotes percentile of node $u$. The weighted error increases linearly with the percentile and decreases with the network size, if the absolute error is constant.
\end{enumerate}

\subsection{Results and Discussion}

In this section, we will discuss simulation results of all proposed methods. The network parameters like size, average, maximum, and minimum degree are estimated using the methods discussed in Section 2. The network size estimation method converges approximately at $1\%$ samples. Each experiment is repeated 10 times and the average value is considered for further experiments.

To measure the performance of the proposed methods, the average error is calculated for each degree and it is averaged over all degrees to compute the overall error in rank estimation. Each value (absolute and weighted error) is computed by taking the average of 20 iterations of the experiment. Results for US, MH, and RW methods are shown when $1\%$ nodes are sampled. All methods are validated on 20 real world social networks, and the summarized results are shown in Table~\ref{avgerr}. The detailed results are shown in Appendix~\ref{appendix1} 

\begin{table}[h!]
\centering
\caption{Average Estimation Error on 20 Real World Social Networks}
\label{avgerr}
\begin{tabular}{|l|l|l|}
\hline
Method & $Err_{paae}$ & $Err_{wtd}$ \\ \hline
PL & 1.51 & 1.14  \\ \hline
US & 0.13 & 0.12  \\ \hline
MH & 0.50 & 0.41  \\ \hline
RW & 0.16 & 0.13  \\ \hline
\end{tabular}
\end{table}

Results show that US method performs best on real world networks. As uniform sampling is not possible in real world networks, RW method is the most feasible and accurate method. In random walk samples, the probability of sampling a node is directly proportional to its degree once the samples are stabled. But in our experiments, we have not removed samples before mixing time and results are shown for the starting $1\%$ samples. It makes the proposed random walk method even faster. MH method gives more error than both US and RW methods because MH random walk is not able to generate perfect uniform samples for small sample size. The efficiency of the sampling methods increases with the sample size. The performance of PL method is poor on real world networks as they do not follow the perfect power law.

\textbf{RW versus PL Method:} In RW method, same samples can be used to estimate the network size and degree rank, so it is faster than the other sampling methods. RW method is also faster than PL method because PL method estimates the average degree using smoothed random walk that is not required in RW method. But if the network parameters are already known, PL method can be used to estimate the rank in $O(1)$ time.

The detailed results are shown for 10 networks and their estimated parameters are shown in Table \ref{est-param}. The average error is shown using actual parameters (A.P.) as well as estimated parameters (E.P.) to observe the error caused due to the estimation of network parameters. Figures~\ref{baerror} and \ref{rwerror} show percentage average absolute error and average weighted error for BA and real world networks respectively. In these figures, first 2 bars show the average error of PL method using actual and estimated parameters respectively. Then it is followed by US, MH, and RW methods. It can be observed that in BA networks, RW method outperforms all other methods. In real world networks, the accuracy of RW method depends on the density and structure of the network. It gives more accurate results in sparse networks than the dense networks. The Same pattern is observed in MH method, as it also collects samples using a walk over the network.

\begin{table}[h!]
\centering
\caption{Estimated Network Parameters}
\label{est-param}
\resizebox{0.9\columnwidth}{!}{%
\begin{tabular}{|l|l|l|l|l|l|l|}
\hline
Network & \multicolumn{2}{|c|}{Number of Nodes} & \multicolumn{2}{|c|}{Average Degree} \\ \hline 
 & Actual & Estimated & Actual & Estimated \\ \hline
BA1 & 100000 & 106773 & 20.00 & 19.68 \\ \hline
BA2 & 200000 & 199303 & 20.00 & 19.75 \\ \hline
BA3 & 300000 & 292649 & 20.00 & 191.09 \\ \hline
BA4 & 400000 & 406837 & 20.00 & 20.06 \\ \hline
BA5 & 500000 & 500688 & 20.00 & 20.30 \\ \hline
Actor & 374511 & 417560 & 80.18 & 92.53 \\ \hline
DBLP & 317080 & 315587 & 6.62 & 7.20 \\ \hline
Digg & 261489 & 260435 & 11.75 & 17.00 \\ \hline
Eu-Email & 224832 & 223151 & 3.02 & 2.96 \\ \hline
Gowalla & 196591 & 199568 & 9.67 & 10.92 \\ \hline
Youtube & 1134885 & 1136445 & 5.26 & 10.15 \\ \hline
\end{tabular}
}
\end{table}

\begin{figure*}[]
  \centering
  \subcaptionbox{Absolute Error}[1.0\linewidth][c]{%
    \includegraphics[width=1.0\linewidth]{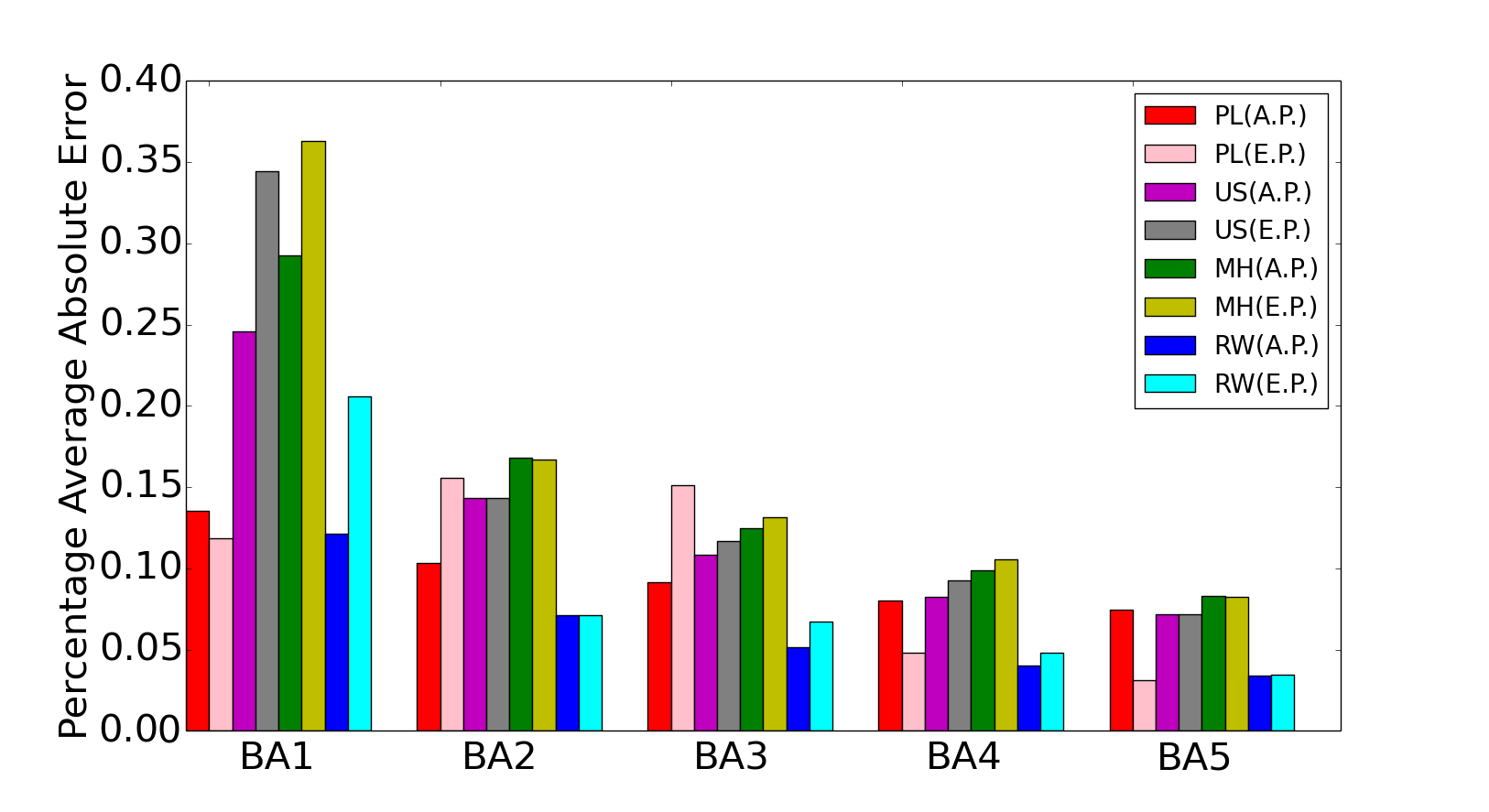}}\quad
  \subcaptionbox{Weighted Error}[1.0\linewidth][c]{%
    \includegraphics[width=1.0\linewidth]{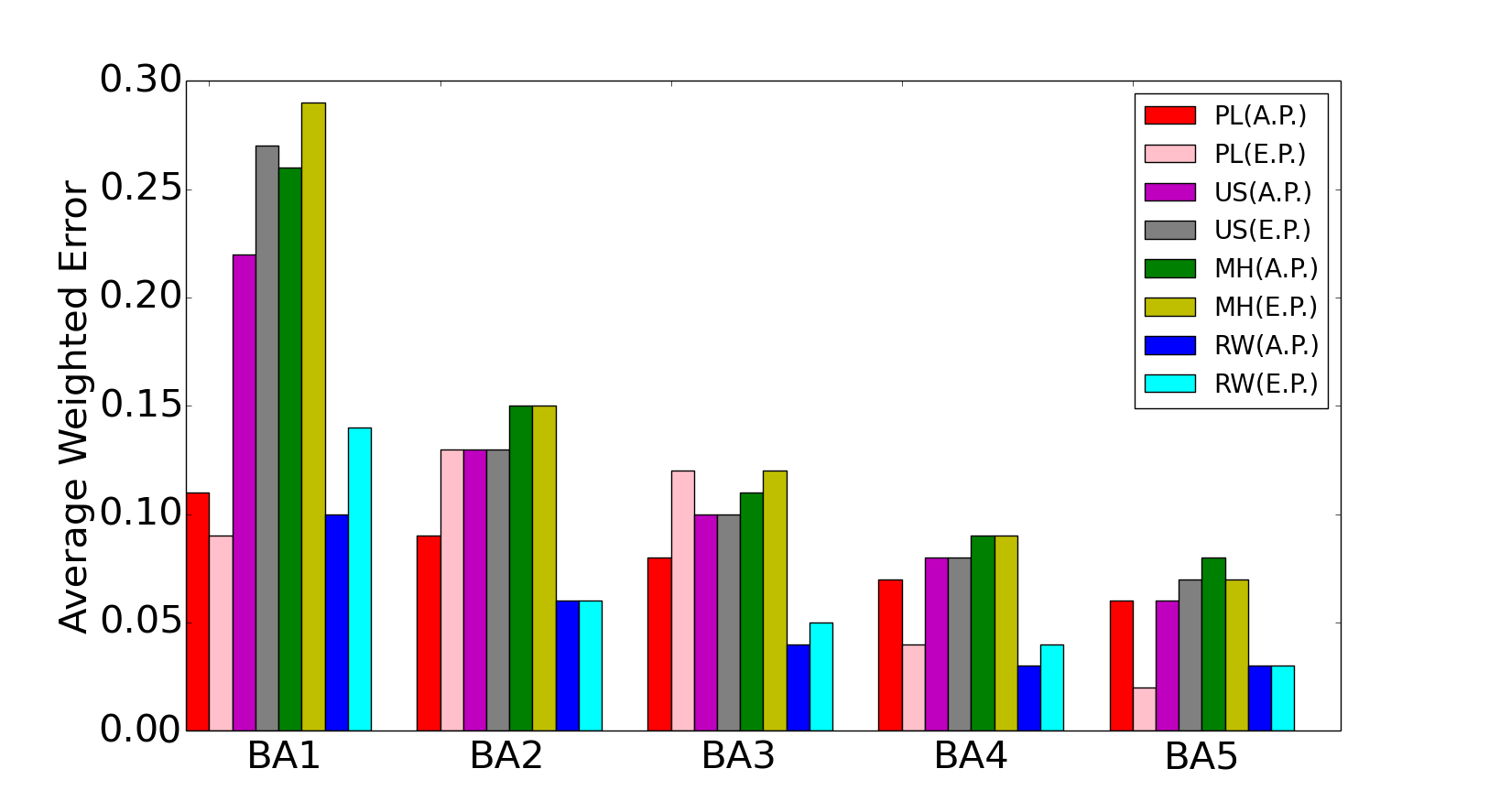}}
  \caption{Average Estimation Error for BA Networks}
  \label{baerror}
\end{figure*}

\begin{figure*}[]
  \centering
  \subcaptionbox{Absolute Error}[1.0\linewidth][c]{%
    \includegraphics[width=1.0\linewidth]{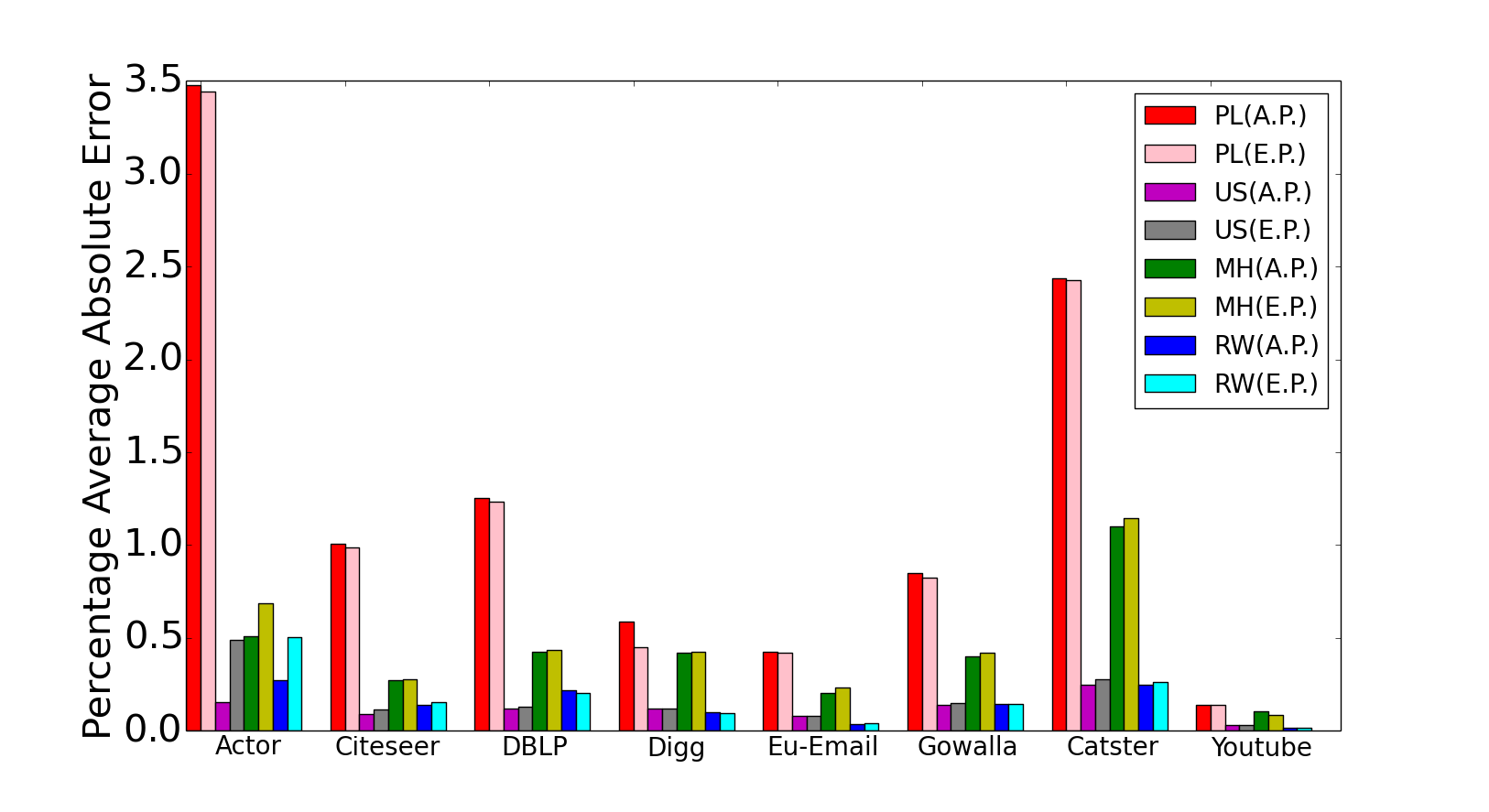}}\quad
  \subcaptionbox{Weighted Error}[1.0\linewidth][c]{%
    \includegraphics[width=1.0\linewidth]{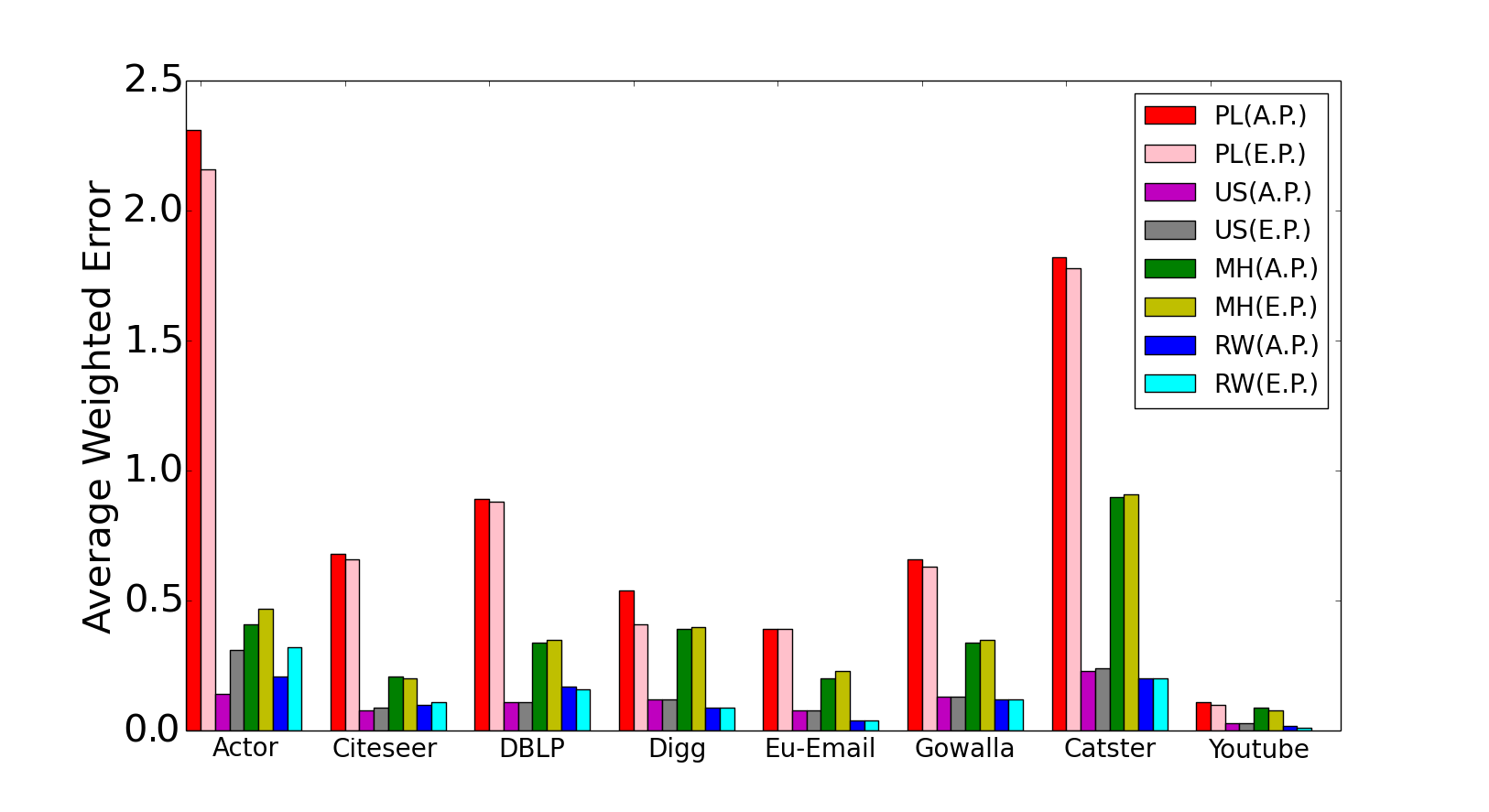}}
  \caption{Average Estimation Error for Real World Networks}
  \label{rwerror}
\end{figure*}

We further study, the behavior of estimation error with degree rank. Figure~\ref{rankerror} shows absolute error versus degree rank for real world (a. Actor, and b. DBLP) networks and BA network. In all methods, estimation error increases with the rank. Figures \ref{rankerror}(a) and \ref{rankerror}(b) show that for high ranked nodes, RW method outperforms other methods. US method gives more error for high ranked nodes, as the linear extrapolation technique starts assigning rank from $n/s$. If $\alpha$ percent nodes are sampled, it will start ranking nodes from $100/\alpha$, that will induce more error for high ranked nodes. Figure \ref{rankerror}(c) shows that PL method outperforms other methods in BA networks, that is followed by RW method. PL method estimates the rank in $O(1)$ time once the preprocessing steps are done. We observe that PL method works well for BA networks but it gives a huge error for real world networks. It happens because of these two reasons. Firstly, real world networks do not follow the perfect power law. Secondly, the rank of a $d$ degree node is computed by integrating the probability distribution function from $d+1$ to $d_{max}$, but in real world networks nodes of some degrees might not be present. This further adds up to more error. 

\begin{figure*}[]
  \centering
  \subcaptionbox{Actor Network}[0.76\linewidth][c]{%
    \includegraphics[width=0.76\linewidth]{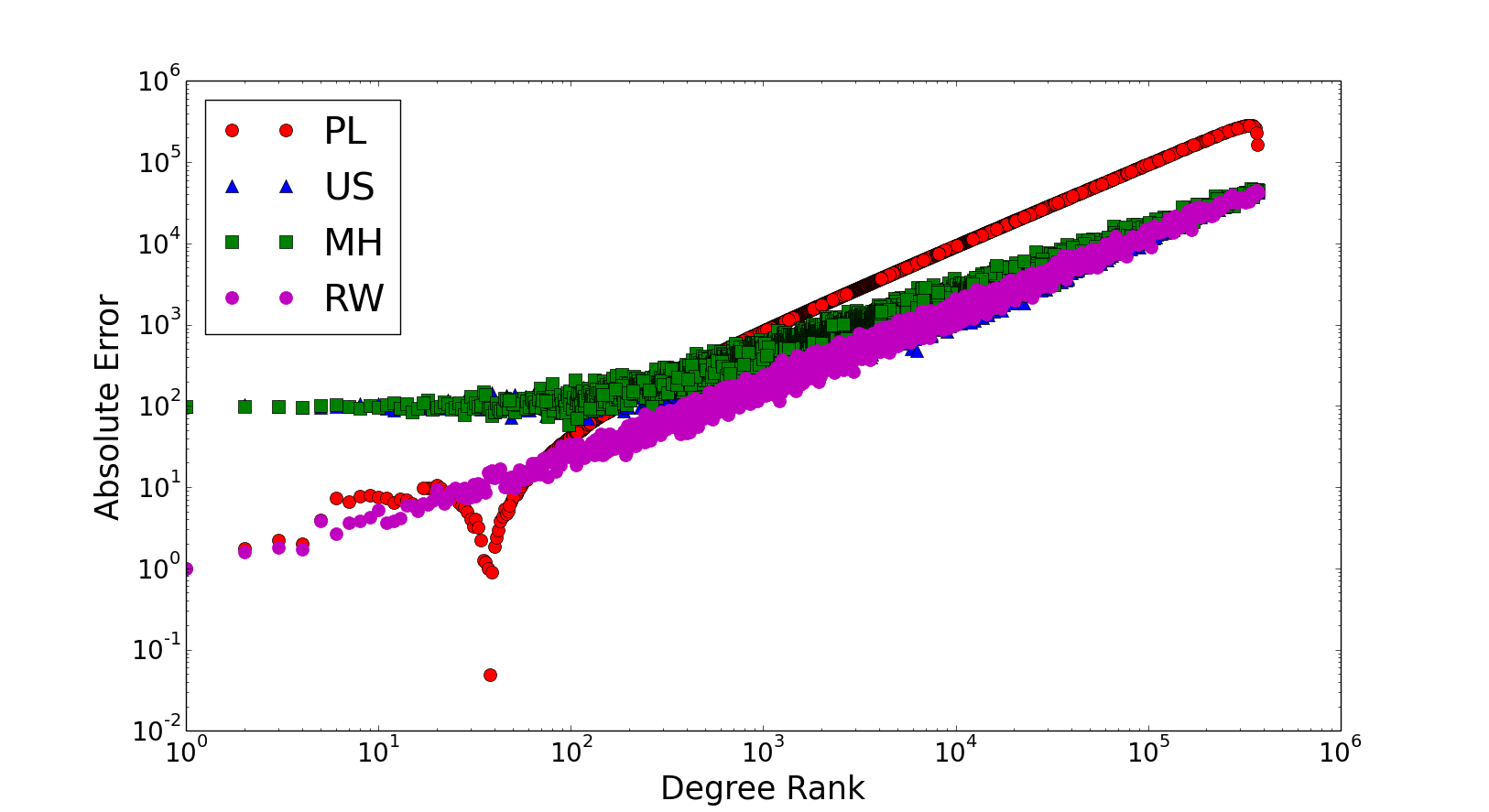}}\quad
  \subcaptionbox{DBLP Network}[0.76\linewidth][c]{%
    \includegraphics[width=0.76\linewidth]{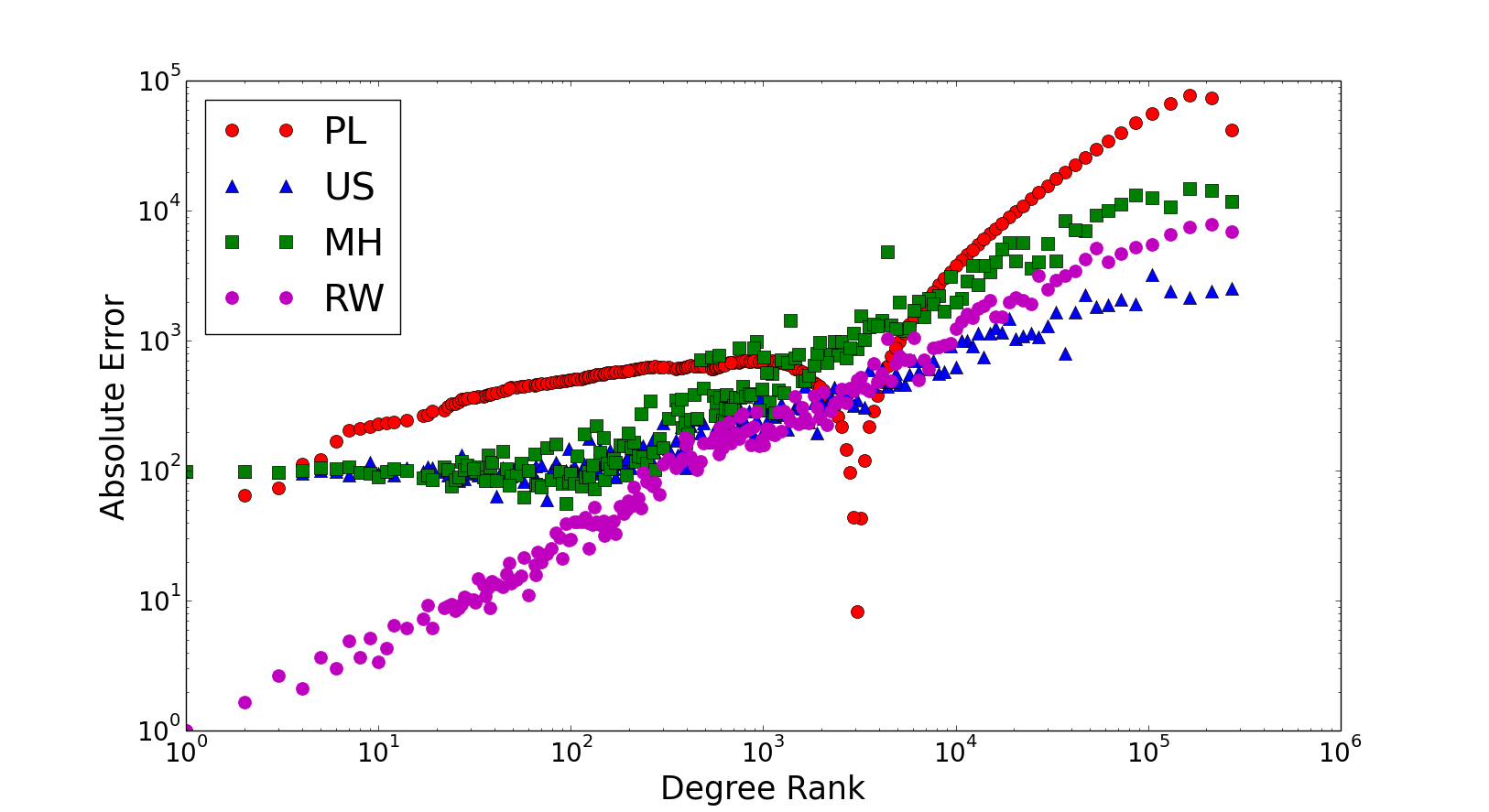}}\quad
    \subcaptionbox{BA1 Network}[0.76\linewidth][c]{%
    \includegraphics[width=0.76\linewidth]{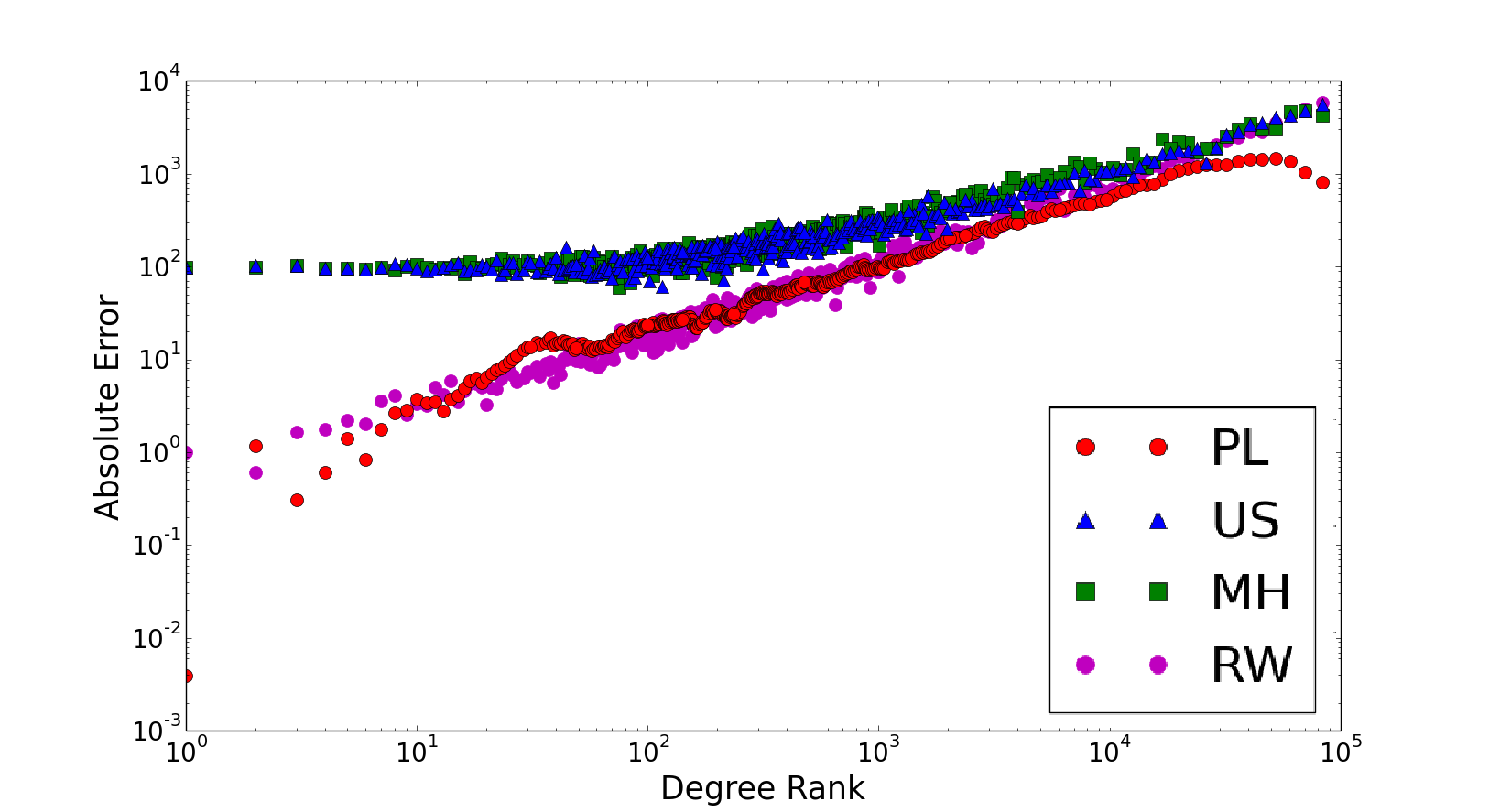}}\quad
  \caption{Absolute Estimation Error versus Degree Rank on log-log scale}
  \label{rankerror}
\end{figure*}

In figure \ref{rankerror}(a) and \ref{rankerror}(b), it can be observed that the rank estimation error in PL method decreases, goes very close to zero, and it further increases. This gives a dip in the absolute error when it is plotted with degree rank. This happens due to the error in the estimated slope of the degree distribution. The actual and estimated number of nodes versus degrees are shown in figure~\ref{degdistest} for DBLP network. So in this network, first, the estimated rank will be lower than the actual rank, then, it will be close to zero when the estimated rank is approximately equal to the actual rank, and finally, the estimated rank will be higher than the actual rank. Due to this reason, it shows a strange dip in the absolute error.

\begin{figure}[htp]
\centering
\includegraphics[width=12cm]{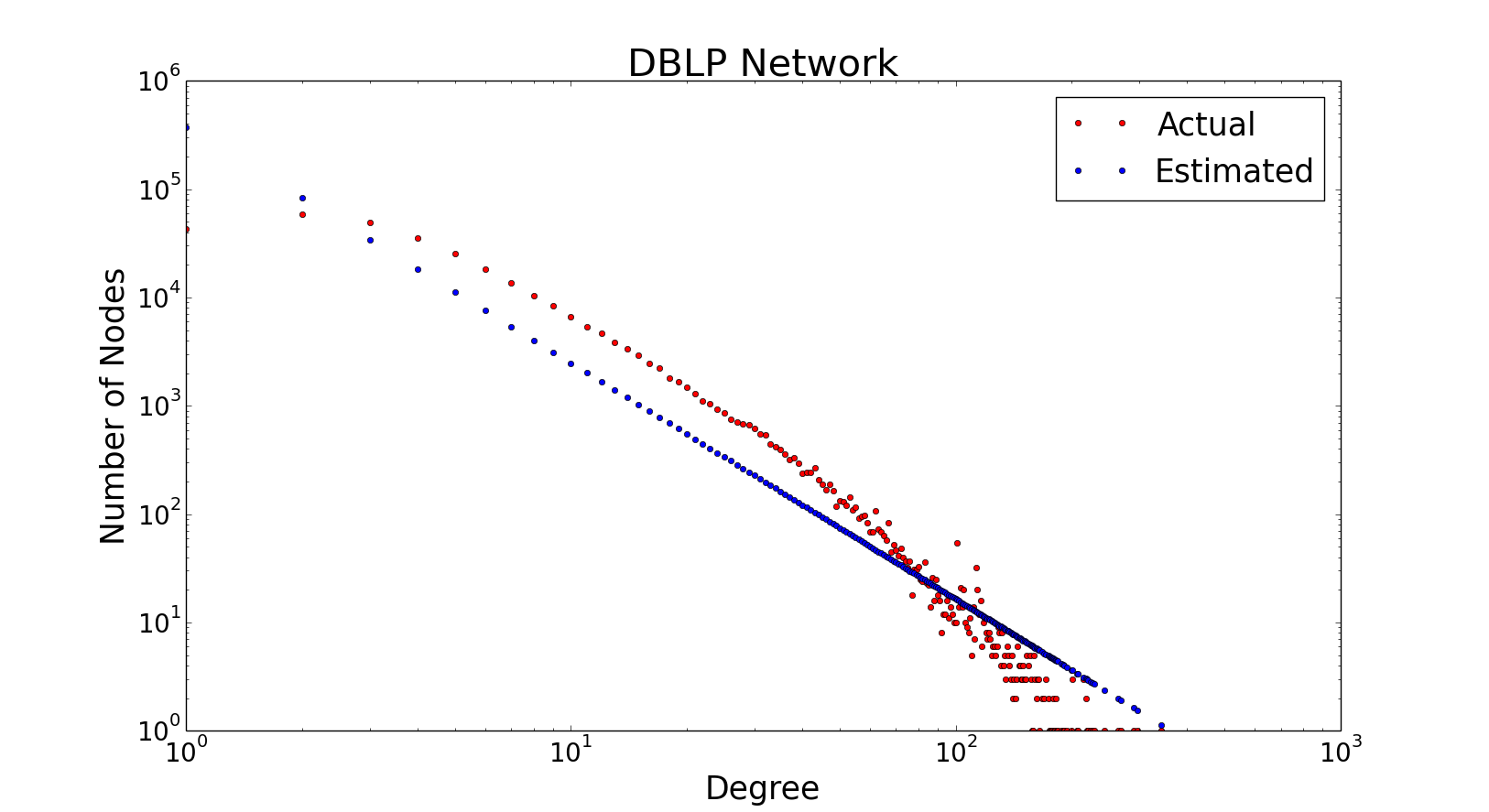}
\caption{Actual and Estimated Number of Nodes versus Degree}
\label{degdistest}
\end{figure}

Next, we study how the estimation error changes with the network size. To study the same, a BA network is evolved by maintaining the same density. In Figure~\ref{nwsize_err}, percentage average absolute error and average weighted error are plotted against the network size. Plots show that the error decreases with an increase in the network size. It can also be observed that RW method outperforms other methods as the network size increases.

\begin{figure*}[]
  \centering
  \subcaptionbox{Absolute Error}[1.0\linewidth][c]{%
    \includegraphics[width=1.0\linewidth]{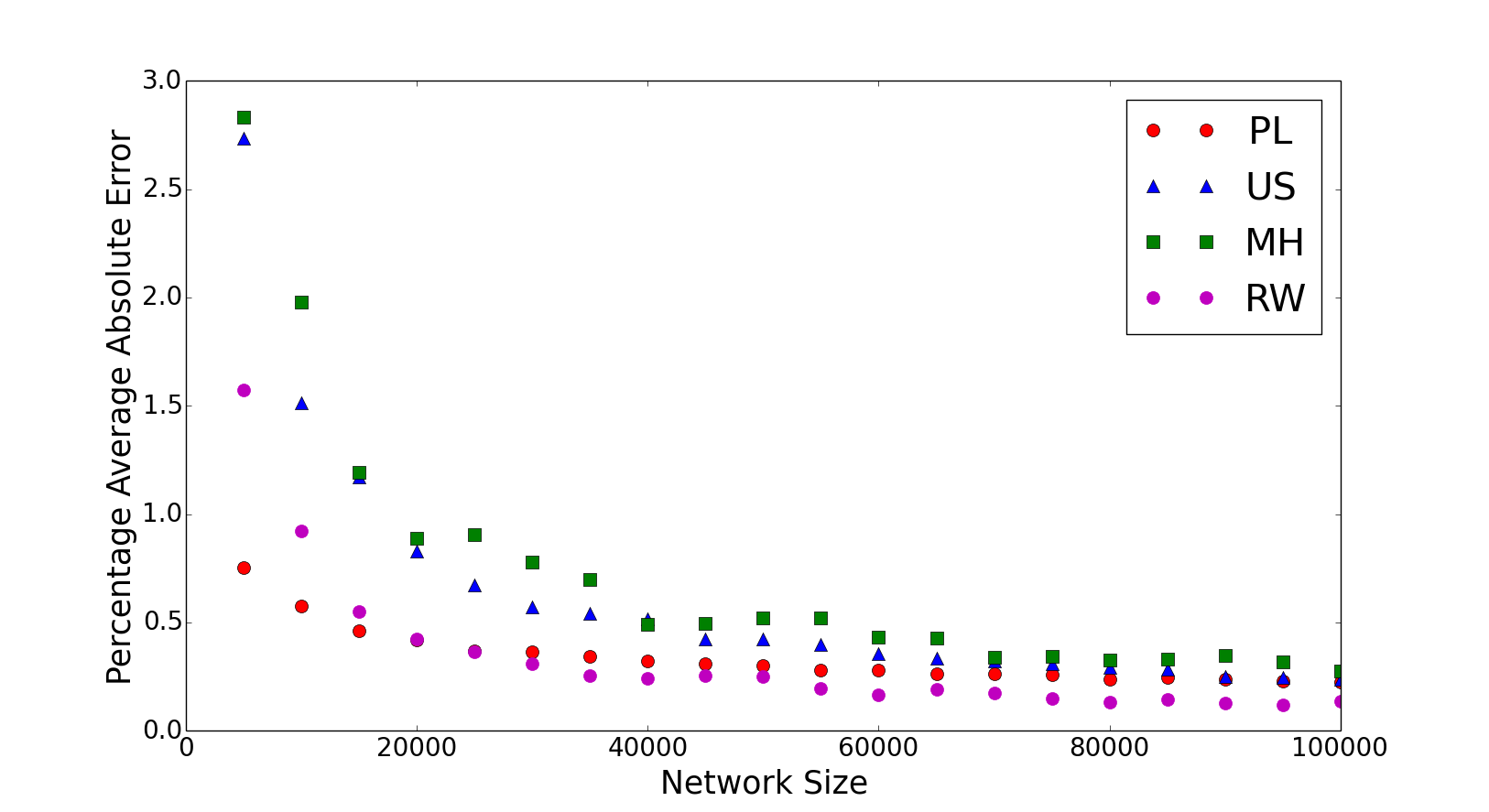}}\quad
  \subcaptionbox{Weighted Error}[1.0\linewidth][c]{%
    \includegraphics[width=1.0\linewidth]{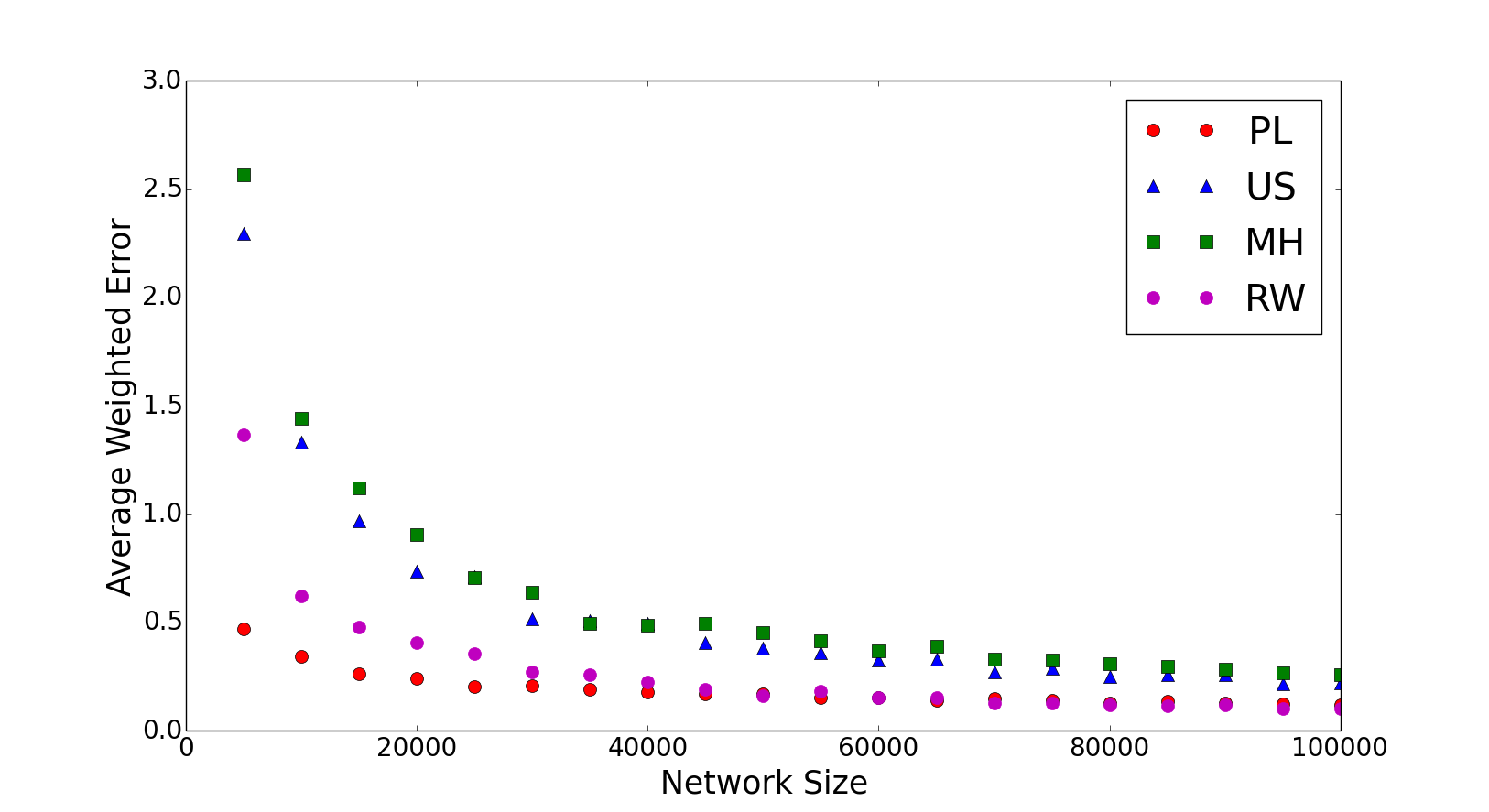}}
  \caption{Average Estimation Error versus Network Size for BA networks}
  \label{nwsize_err}
\end{figure*}

\section{Random Networks}

In this section, we study degree ranking methods for random networks. In 1959, Erd\H{o}s and R\'{e}nyi proposed a model to generate random networks, called Erd\H{o}s and R\'{e}nyi (ER) model \cite{erd6s1960evolution}. In ER model, a network is started with $n$ nodes, and an edge is placed between each pair of nodes with some fixed probability $p$. The degree distribution $g$ of random networks follows poisson law. The probability of a node having degree $j$ can be approximated as $g(j) \rightarrow \frac{(d_{avg})^je^{-d_{avg}}}{j!}$ as $n \rightarrow \infty$, where $d_{avg}$ is average degree of the network. We will discuss a ranking method based on poisson law degree distribution. The rest of the methods (US, MH, and RW method) can be directly applied to random networks, as they have no dependency on the degree distribution function. Network size and average degree are estimated using the same techniques that we have discussed earlier. 

\begin{table}[]
\centering
\caption{Estimated parameters for Erd\H{o}s and R\'{e}nyi Networks}
\label{summ}
\resizebox{\columnwidth}{!}{%
\begin{tabular}{|l|l|l|l|l|l|l|}
\hline
Network & \multicolumn{2}{|c|}{Number of Nodes} & \multicolumn{2}{|c|}{Average Degree} \\ \hline 
 & Actual & Estimated & Actual & Estimated \\ \hline
ER1 & 100000 & 99874 & 11.50 & 11.24 \\ \hline
ER2 & 200000 & 202731 & 12.34 & 12.08 \\ \hline
ER3 & 300000 & 300503 & 12.71 & 12.49 \\ \hline
ER4 & 400000 & 398168 & 12.99 & 121.01 \\ \hline
ER5 & 500000 & 505675 & 13.19 & 13.07 \\ \hline
\end{tabular}
}
\end{table}

\subsection{Using Poisson Degree Distribution (PD Method)}

This method uses Poisson degree distribution of random networks to estimate degree rank of a node.

\begin{lemma}
In a random network G $(G \in \mathcal{G}(g))$, the expected degree rank of a node $u$ can be computed as, $E[R_{G}(u)]= n \cdot e^{-d_{avg}}\sum_{j=d_u+1}^{d_{max}}\frac{(d_{avg})^j}{j!} +1$.
\end{lemma}

\begin{proof}

In a given network $G$ that follows Poisson degree distribution, the actual rank of a node $u$ can be computed as, 

\begin{center}
$R_{act}(u) = \sum_{j=d_{u}+1}^{d_{max}}n_j  + 1$
\end{center}
where, $n_j$ represents total number of nodes having degree $j$ in network $G$. 

Let $N_j$ be a random variable that represents the total number of nodes having degree $j$ in the network $G$ $(G \in \mathcal{G}(g))$. The expected value of $N_j$ is $E[N_j]=n\cdot g(j)$. 
Then, the expected degree rank of a node $u$ can be computed as,
\begin{align*}
E[R_{G}(u)] &= E\left[ \sum_{j=d_{u}+1}^{d_{max}}N_j  + 1\right] \\
E[R_{G}(u)] &=  \sum_{j=d_{u}+1}^{d_{max}}E[N_j]  + 1\\
E[R_{G}(u)] &= \sum_{j=d_{u}+1}^{d_{max}}n\cdot g(j)  + 1
\end{align*}

As we know $g(j) \rightarrow \frac{(d_{avg})^je^{-d_{avg}}}{j!}$ as $n \rightarrow \infty$, so to compute the expected rank we use $g(j) = \frac{(d_{avg})^je^{-d_{avg}}}{j!}$,
\begin{align*}
E[R_{G}(u)] &=  n \cdot \sum_{j=d_{u}+1}^{d_{max}} \frac{(d_{avg})^je^{-d_{avg}}}{j!} +1\\
E[R_{G}(u)] &= n \cdot e^{-d_{avg}}\sum_{j=d_u+1}^{d_{max}}\frac{(d_{avg})^j}{j!} +1,
\end{align*}
as desired.~\end{proof}

\begin{corollary}
In a random network $G$, the degree rank of a node $u$ can be estimated as,\\ $R_{est}(u)= n \cdot e^{-d'_{avg}}\sum_{j=d_u+1}^{d'_{max}}\frac{(d'_{avg})^{j}}{j!} +1$, where $d'_{max}$ and $d'_{avg}$ are estimated maximum and average degree of the network respectively.
\end{corollary}

\subsection{Discussion}

The proposed methods (PD, US, MH, RW) are verified on the generated ER networks, and their details are given in table~\ref{summ}. Figure~\ref{ernw} shows percentage average absolute error and average weighted error using actual and estimated parameters.
In PD method, the error computed using estimated parameters is very high as the rank is highly dependent on the average degree. The number of nodes for each degree $j$ is directly proportional to $(d_{avg})^j$, so, a small estimation error leads to more cumulative error. 
Rest of the results are similar to scale-free networks. The average error of RW method is very close to US method, and it can be efficiently used for large size random networks. The performance of RW method improves with network size. It is also observed that the estimation error decreases as the network size increases. In random networks, absolute error versus degree rank shows a different behavior due to the poisson degree distribution. Figure~\ref{abserrer} shows that the estimation error first increases with the rank and then decreases.

\begin{figure*}[]
  \centering
  \subcaptionbox{Absolute Error}[1.0\linewidth][c]{%
    \includegraphics[width=1.0\linewidth]{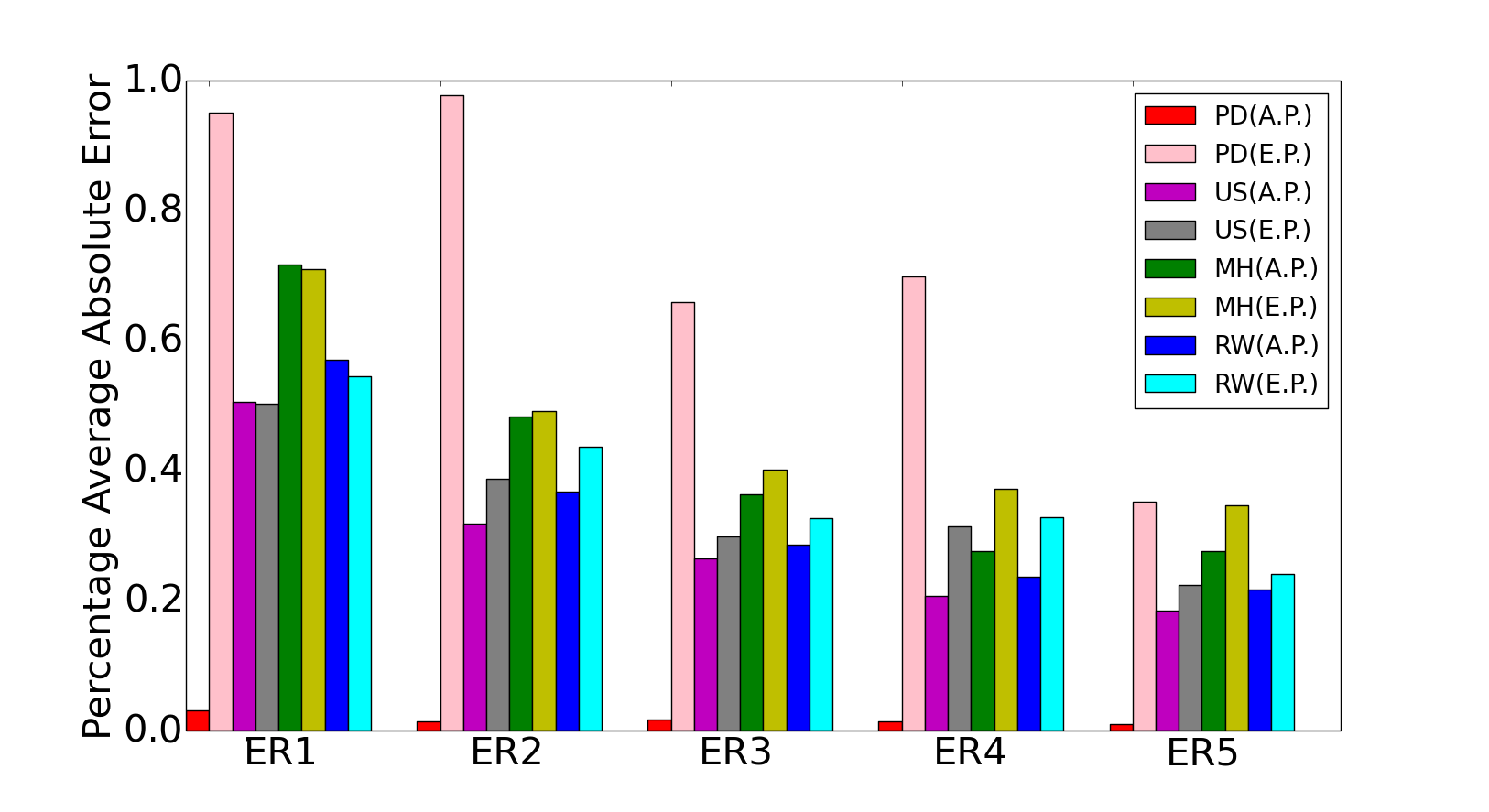}}\quad
  \subcaptionbox{Weighted Error}[1.0\linewidth][c]{%
    \includegraphics[width=1.0\linewidth]{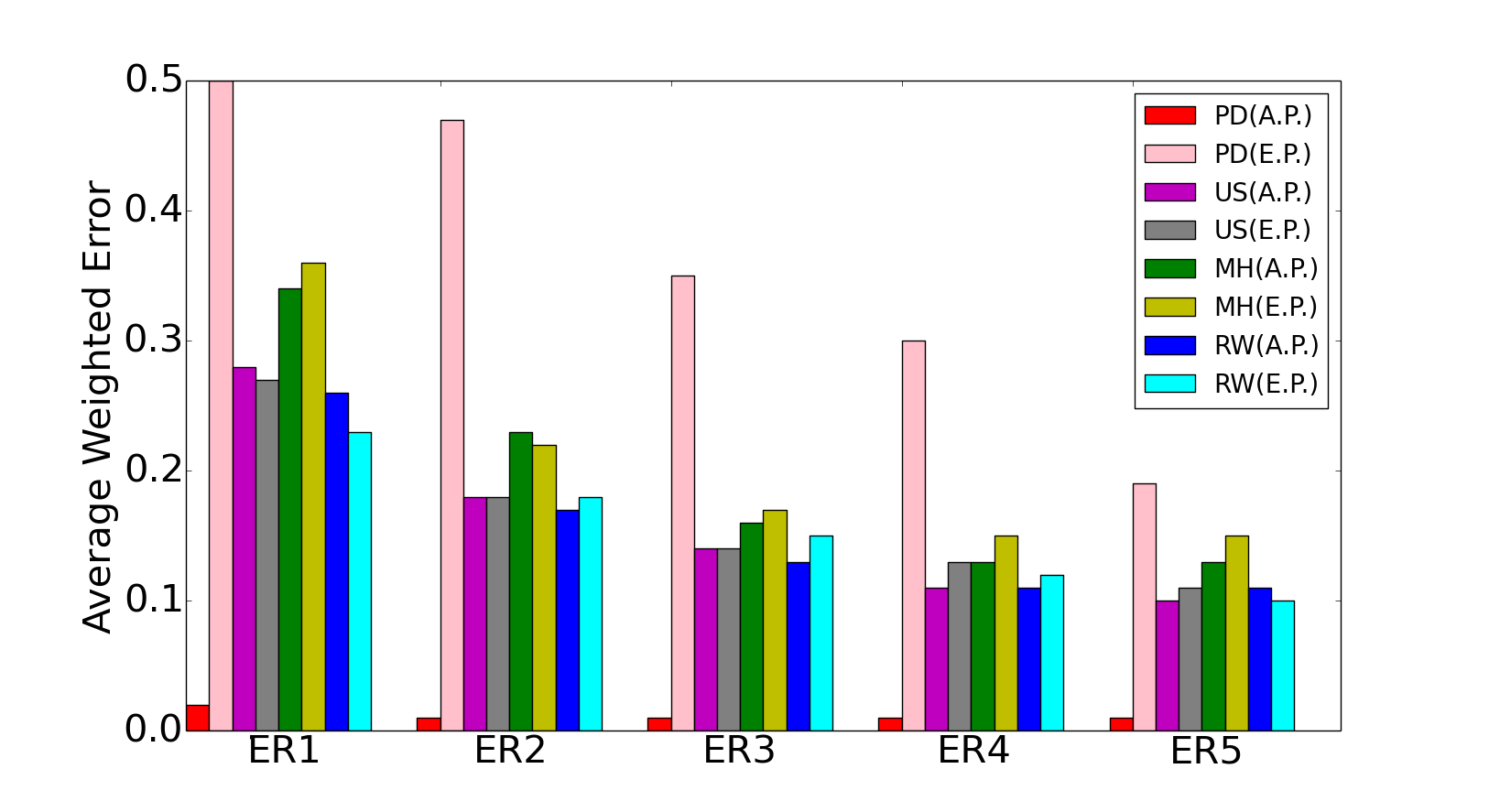}}
  \caption{Average Estimation Error for ER Networks}
  \label{ernw}
\end{figure*}

\begin{figure}[]
\centering
\includegraphics[width=12cm]{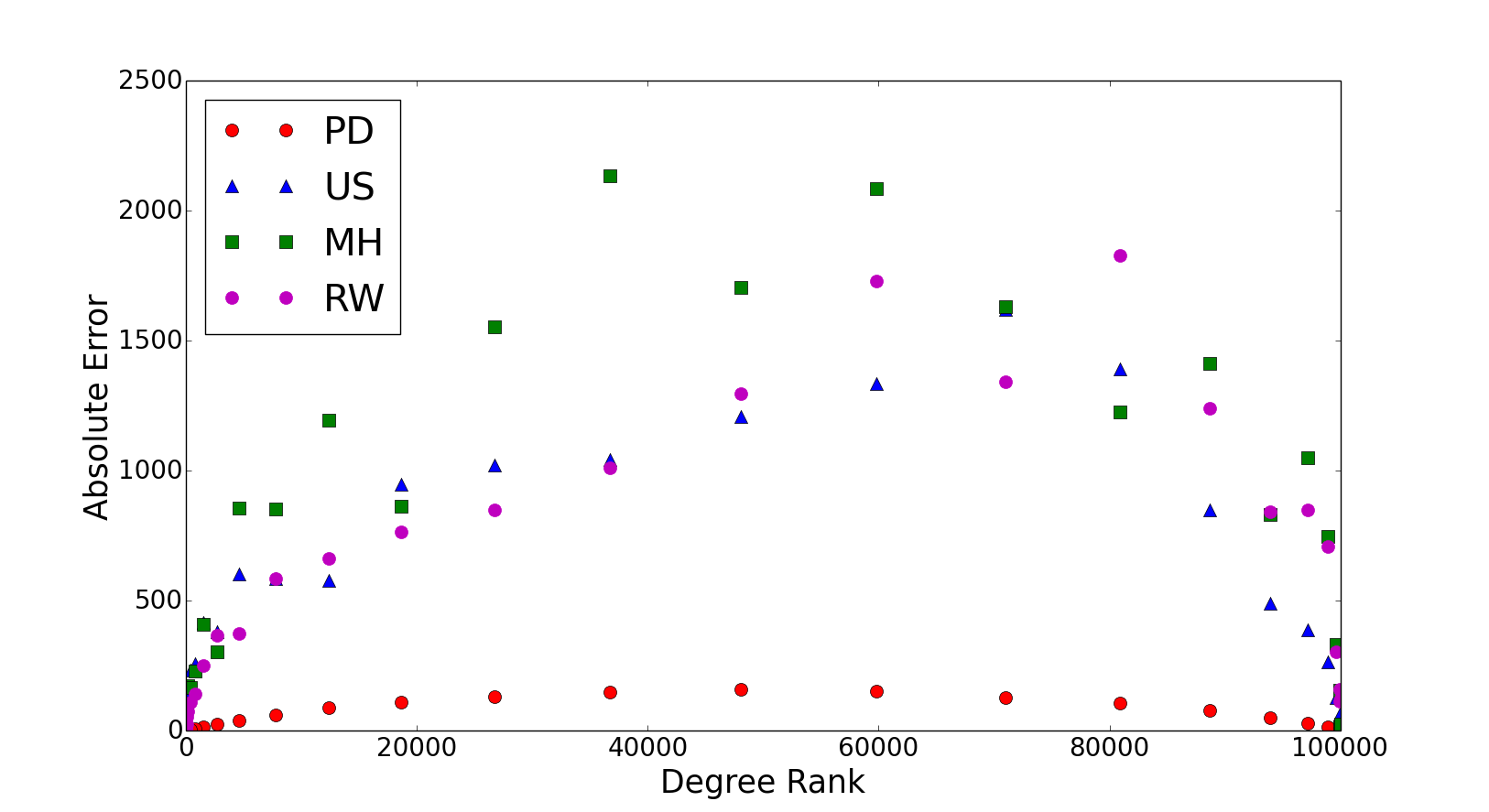}
\caption{Absolute Estimation Error versus Degree Rank for ER1 Network}
\label{abserrer}
\end{figure}

\section{Conclusion}

In this work, we have proposed four methods to estimate degree rank of a node without having the entire structure of the network. With time, the size of real world dynamic networks is increasing very fast. It is not feasible to collect the entire network to study its global properties. The proposed methods collect a small sample set using random walk or its variations and estimate global degree rank of the node. 

The accuracy of the proposed methods is evaluated using absolute and weighted error functions. It is observed that the accuracy of RW method is very close to US method. RW method is the most feasible and accurate method for real world networks. In RW method, percentage average absolute error is $0.16\%$ and average weighted error is $0.13\%$. All proposed methods estimate the rank of higher degree nodes more accurately than the lower degree nodes. These methods are further extended to random networks. Results show that they can be efficiently used to estimate degree rank in random networks.


One can further extend this to estimate the rank of a node based on other centrality measures like closeness centrality, betweenness centrality, katz centrality, pagerank, coreness, and so on. The complexity to compute these global centrality measures is very high. So, the local algorithms to compute the global rank of the nodes will be of great interest. 


\bibliographystyle{elsarticle-num}

\bibliography{/home/akrati/latex/mybib}

\section*{Appendix}
\begin{subappendices}
\renewcommand{\thesection}{\Alph{section}}%

\section{Results on real world scale-free networks}\label{appendix1}

\begin{landscape}
\begin{table}[htp]
\centering
\caption{Absolute and Weighted Error in the Estimated Ranking using all Methods on Real World Social Networks}
\label{my-label1}
\begin{tabular}{|l|l|l|l|l|l|l|l|l|l|l|l|l|l|}
\hline
Network & Type & Ref & Nodes & Edges & Avg Deg & \multicolumn{2}{|c|}{PL Error} & \multicolumn{2}{|c|}{US Error} & \multicolumn{2}{|c|}{MH Error} & \multicolumn{2}{|c|}{RW Error} \\ \hline

& & & & &  & Abs & Wt &  Abs & Wt & Abs & Wt & Abs & Wt \\ \hline

Friendster & Social & \cite{ nr} & 5689498 & 14067887 & 4.95 & 0.06 & 0.06 & 0.01 & 0.01 & 0.05 & 0.05 & 0.01 & 0.01 \\ \hline
Academia & Social & \cite{Fire2011} & 200167 & 1022440 & 10.22 & 1.46 & 1.01 & 0.15 & 0.14 & 0.47 & 0.35 & 0.19 & 0.14 \\ \hline
Dogster & Social & \cite{ nr} & 426485 & 8543321 & 40.06 & 1.27 & 0.95 & 0.1 & 0.09 & 0.34 & 0.28 & 0.07 & 0.06 \\ \hline
Facebook1 & Social & \cite{traud2012social } & 3097165 & 23667394 & 15.28 & 01.07 & 0.78 & 0.03 & 0.02 & 0.19 & 0.17 & 0.05 & 0.04 \\ \hline
Gowalla & Social & \cite{ cho2011friendship} & 196591 & 950327 & 9.67 & 01.05 & 0.66 & 0.12 & 0.11 & 0.5 & 0.41 & 0.12 & 0.1 \\ \hline
Hyves & Social & \cite{ zafarani2009social} & 1402673 & 2777419 & 3.96 & 0.2 & 0.15 & 0.02 & 0.02 & 0.07 & 0.07 & 0.01 & 0.01 \\ \hline
Foursquare & Social & \cite{ zafarani2009social} & 639014 & 3214985 & 10.06 & 1.1 & 0.95 & 0.08 & 0.07 & 0.52 & 0.45 & 0.06 & 0.05 \\ \hline
Last.fm & Social & \cite{ konstas2009social} & 1191805 & 4519330 & 7.58 & 0.24 & 0.21 & 0.03 & 0.03 & 0.09 & 0.08 & 0.01 & 0.01 \\ \hline
Livemocha & Social & \cite{ ZafaraniLiu2009} & 104103 & 2193082 & 42.13 & 2.96 & 2.26 & 0.42 & 0.38 & 1.01 & 0.78 & 0.31 & 0.24 \\ \hline
Delicious & Social & \cite{nr } & 536108 & 1365961 & 5.1 & 0.31 & 0.25 & 0.05 & 0.05 & 0.17 & 0.15 & 0.05 & 0.04 \\ \hline
Douban & Social & \cite{nr } & 154908 & 327162 & 4.22 & 1.35 & 1.16 & 0.19 & 0.18 & 0.41 & 0.38 & 0.13 & 0.12 \\ \hline
Actor & Collaboration & \cite{barabasi1999emergence} & 374511 & 15014839 & 80.18 & 3.48 & 2.31 & 0.15 & 0.14 & 0.51 & 0.41 & 0.27 & 0.21 \\ \hline
DBLP & Collaboration & \cite{yang2015defining } & 317080 & 1049866 & 6.62 & 1.25 & 01.09 & 0.12 & 0.11 & 0.43 & 0.34 & 0.21 & 0.16 \\ \hline
Digg & Social & \cite{de2009social } & 261489 & 1536577 & 11.75 & 0.59 & 0.54 & 0.12 & 0.12 & 0.42 & 0.39 & 0.1 & 0.09 \\ \hline
Eu-Email & Communication & \cite{leskovec2007graph } & 224832 & 339925 & 3.02 & 0.42 & 0.39 & 0.08 & 0.08 & 0.21 & 0.2 & 0.04 & 0.04 \\ \hline
Gplus & Social & \cite{ mcauley2012learning} & 107614 & 12238285 & 227.45 & 5.93 & 4.64 & 0.46 & 0.41 & 21.08 & 2.34 & 1.09 & 01.07 \\ \hline
Catster & Social & \cite{ nr} & 148826 & 5447464 & 73.21 & 2.43 & 11.02 & 0.25 & 0.23 & 1.1 & 0.9 & 0.25 & 0.2 \\ \hline
Youtube & Social & \cite{ zafarani2009social} & 1134885 & 2987623 & 5.27 & 0.14 & 0.11 & 0.03 & 0.03 & 0.1 & 0.09 & 0.02 & 0.02 \\ \hline
Pokec & Social & \cite{ nr} & 1632803 & 22301964 & 27.32 & 2.74 & 1.78 & 0.05 & 0.04 & 0.14 & 0.1 & 0.07 & 0.05 \\ \hline
Hollywood & Collaboration & \cite{BoVWFI} & 1069126 & 56306653 & 105.33 & 2.54 & 1.91 & 0.08 & 0.07 & 0.3 & 0.26 & 0.15 & 0.12 \\ \hline

Summary &  &  &  &  &  & 1.51 & 1.14 & 0.13 & 0.12 & 0.50 & 0.41 & 0.16 & 0.13 \\ \hline

\end{tabular}
\end{table}
\end{landscape}

\end{subappendices}

\end{document}